\documentclass[preprint,12pt]{elsarticle}

\usepackage{amsmath}
\usepackage{algpseudocode}
\usepackage{algorithm}
\usepackage{subfig}
\usepackage{stfloats}
\usepackage{balance}
\usepackage{mathtools}
\usepackage{epstopdf}

\algnewcommand\algorithmicforeach{\textbf{for each}}
\algdef{S}[FOR]{ForEach}[1]{\algorithmicforeach\ #1\ \algorithmicdo}

\def\BibTeX{{\rm B\kern-.05em{\sc i\kern-.025em b}\kern-.08emT\kern-.1667em\lower.7ex\hbox{E}\kern-.125emX}}

\newtheorem{remark}{Remark}
\newtheorem{assumption}{Assumption}
\newtheorem{proposition}{Proposition}
\newtheorem{lemma}{Lemma}
\newtheorem{proof}{Proof}
\newtheorem{definition}{Definition}

\journal{ACM e-Energy 2020}

\begin{document}
\begin{frontmatter}
%
\title{Energy Storage as Public Asset}

%
\author{Jiasheng Zhang, Nan Gu, and Chenye Wu}

%

%
\begin{abstract}
  Energy storage has exhibited great potential in providing flexibility in power system to meet critical peak demand and thus reduce the overall generation cost, which in turn stabilizes the electricity prices. In this work, we exploit the opportunities for the independent system operator (ISO) to invest and manage storage as public asset, which could systematically provide benefits to the public. Assuming a quadratic generation cost structure, we apply parametric analysis to investigate the ISO's problem of economic dispatch, given variant quantities of storage investment.  This investment is beneficial to users on expectation. However, it may not necessarily benefit everyone. We adopt the notion of marginal system cost impact (MCI) to measure each user's welfare and show its relationship with the conventional locational marginal price. We find interesting convergent characteristics for MCI. Furthermore, we perform $k$-means clustering to classify users for effective user profiling and conduct numerical studies on both prototype and IEEE test systems to verify our theoretical conclusions. 
\end{abstract}

%
%

%
\end{frontmatter}
%

%


\section{Introduction}
One of the key bottlenecks in improving the effectiveness of electricity sectors is the limited flexibility in the power system, which leads to the limited fluidity in the market. Fortunately, over the past few decades, technological improvements together with the scale of economy have significantly reduced the cost of various types of storage systems, and this trend is projected to continue in next years (as shown in Figure \ref{storage}). The storage system, if widely deployed, can provide the urgently needed flexibility to the power system, which will dramatically relieve the pressure in electricity market design. For example, it can relieve the critical peak in the system \cite{ugarte2015energy}, and mitigate too much uncertainties brought by the renewables \cite{denholm2011grid}. These and other services that storage system provides to the grid can benefit both the system operator (the system as a whole) as well as individual consumers. While most researches focus on incentivizing individual storage owners to provide services to the grid, we consider an alternative to view the storage as public asset. In essence, widely deployed storage system, just as most publicly-owned infrastructures in the grid, requires huge investment, yet it can generate comparable economic value with potential long-term returns.
\begin{figure}[t]
        \centering  
        \includegraphics[width=0.95\linewidth]{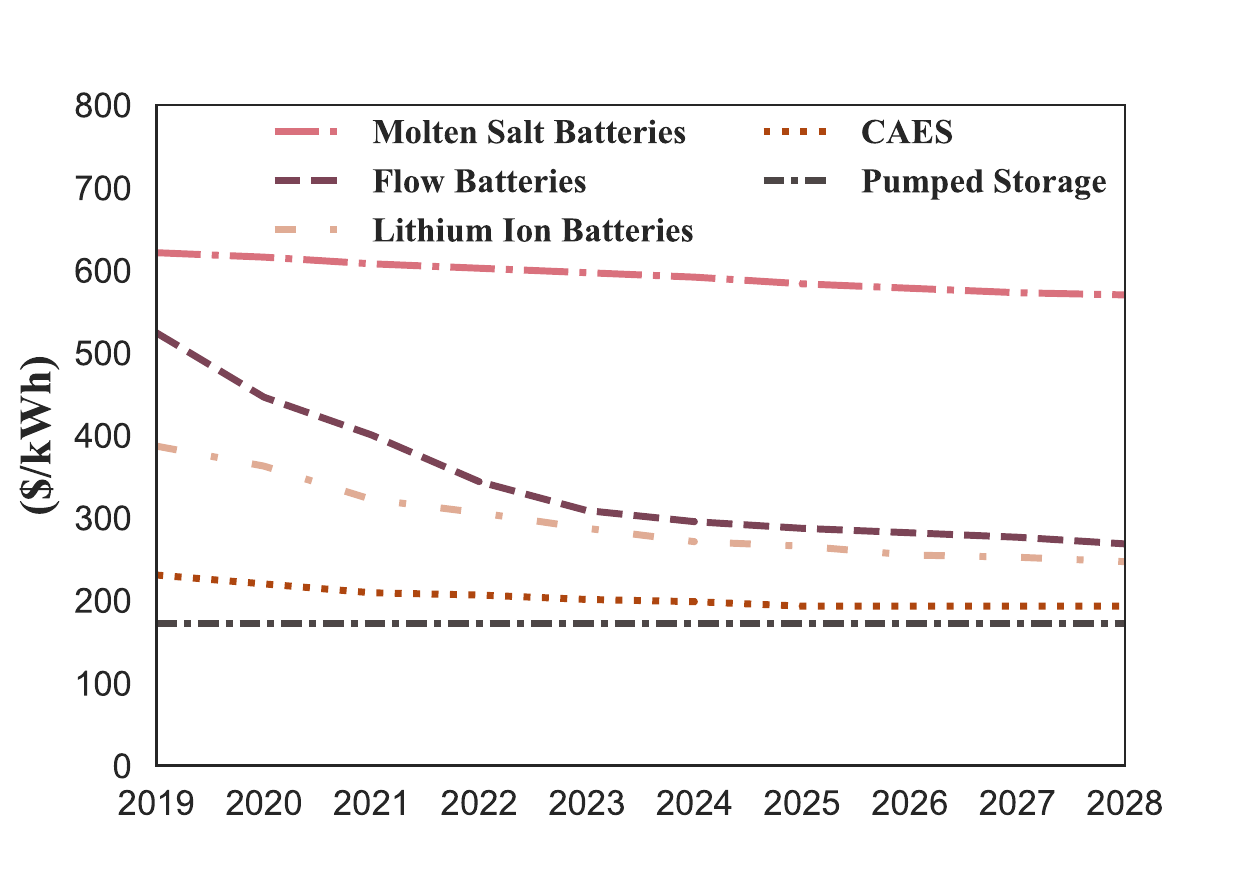}
        \caption{Projected Diminishing Marginal Costs for Variant Storage Technologies \cite{Navigant2019}.} 
        \label{storage}
        \vspace{-0.1cm}
\end{figure}
However, the large-scale deployment of storage could pose new challenges to the electricity market design. The major difficulty is exactly due to the large-scale deployment. In this case, storage systems can no longer be viewed as price-takers and will have a major impact on the current locational marginal price (LMP) scheme. At first glance, one may believe the storage system could help reduce the electricity bills for all users. This intuition is wrong. The truth is that the storage system could only help reduce the "average" electricity price over time and across  all the locations. This smoothing effect will of course benefit some market participants but make other participants worse-off. In this paper, we exploit how the integration of storage system will change the definition of conventional LMP, which serves as the basis for us to understand users in terms of their potential benefits. This also allows us to conduct $k$-means clustering to better distinguish heterogeneous users in the new market conditions. 


Moreover, we characterize the smoothing effect rigorously by examining the global convergence of the LMP scheme as storage capacity increases. We could in turn reason the dynamics of individual electricity bills as the total storage capacity in the grid increases. {We respectively highlight the impacts of publicly owned storage in two models: electricity pool model and network constrained model. The results of the former case can be applied in the micro-grid scenario and the latter emphasizes the effects of grid interchanges. }

\subsection{Related Works}
Our work roots in two research lines: the electricity storage control framework design and the pricing mechanism investigation in electricity market.

While storage control framework design has been well investigated, most researches either focus on individually owned storage control policy design (e.g., to conduct arbitrage against Time-of-Use (ToU) prices, or real time prices) or consider a central control framework in various electricity operation processes. For example, Tang \textit{et al.} discuss the dispatch game between independent system operator (ISO) and generator-owned storage in \cite{tang2015dynamic}. Bose \textit{et al.} show the variability and the locational marginal value of generator-owned energy storage in \cite{bose2014variability}. Mohsenian-Rad \textit{et al.} propose a framework to coordinate the investor-owned storage facilities in power system in  \cite{mohsenian2015coordinated}. Cui \textit{et al.} further the research by considering wind power integration in \cite{cui2017bilevel}. In \cite{lakshminarayana2016cooperation}, Lakshminarayana \textit{et al.} devise an operation schedule to centrally coordinate multiple storage devices. Qin \textit{et al.} design an algorithm to use storage to mitigate the uncertainties brought by renewables in \cite{qin2015distributed}. Grillo \textit{et al.} employ a Markov decision process to determine the optimal storage scheduling policy with time-varying renewable generation in \cite{grillo2015optimal}. Wang \textit{et al.} propose a dynamic programming algorithm for storage users' arbitrage scheme against multi-peaked ToU pricing in \cite{wang2019algorithmic}. Xu \textit{et al.} present an optimal look-ahead storage control policy for arbitrage based on Lagrangian multipliers in \cite{xu2019lagrangian}. Different from this line of research, we consider the storage system as public asset and examine both its value to the system operator and its benefit to individual market participants (through LMP analysis). Specifically, we use parametric analysis to exploit the value of storage. This technique has been utilized to understand the relationship between ramping capacities and overall generation cost in \cite{wu2015risk}. 
Parametric convex quadratic optimization is discussed in detail in \cite{kheirfam2010sensitivity}, \cite{Romanko04aninterior}. 


The pricing mechanism for the electricity sector has also caught much attention. Xu \textit{et al.} adopt VCG mechanism to design incentive compatible pricing scheme in \cite{xu2015efficient}. Kim \textit{et al.} in \cite{kim2015dynamic} propose a market scenario where both utility companies and customers employ reinforcement learning strategies to determine real-time price and schedule energy consumptions. Specifically, in the field of LMP, Oren \textit{et al.} clarify the definition of LMP 
and analyze the role of transmission rights on LMP in \cite{oren1995nodal}. Li \textit{et al.} address the step change issue of LMP when load variation occurs and raise a new continuous solution to this issue in \cite{li2007continuous}.
Bai \textit{et al.} redefine LMP in a market with various forms of distributed energy resources and decompose it into several components according to the physical attributes in \cite{bai2017distribution}. In contrast to the previous works, we identify that LMP, as its name suggests, is designed to exploit spatial features. However, storage system introduces temporal coupling into the pricing scheme, which warrants a re-consideration on the definition of LMP. Cui \textit{et al.} analyze the smoothing effect for LMP by storage in \cite{cui2017bilevel}, which is closely related to our topic. In contrast, we apply a data-driven approach to enable customized pricing schemes. This approach has been discussed in \cite{yu2017good}, where Yu \textit{et al.} classify user types to identify their economic information. This inspires our thought on measuring users' marginal impact when storage is deployed as a public asset. {Another distinct difference is that we focus on the operation of storage, so the investment cost of storage is not considered.}
\subsection{Our Contributions}
In seek of exploiting the value of storage system as public asset to the grid, the principal contributions of our work can be summarized as follows:
\begin{itemize}
  \item \textit{LMP Scheme with Storage}: We exploit the definition of LMP with storage system as public asset, and decompose it in terms of spatial components (conventional definition) and temporal components (new components induced by storage).
  \item \textit{Storage's Impact on System}: We prove storage helps to increase social welfare. Besides, we characterize the smoothing effect induced by LMP, both in the electricity pool model and in the general network constrained model. This highlights the value of storage as public asset to the system as a whole.
  \item \textit{Data-driven User Profiling}: The new definition of LMP can help us characterize the users with big data. Such user profiling enables us to examine the value of storage as public asset for individuals. Specifically, the marginal system cost impact (MCI) for different kinds of users at the same bus tend to converge when storage capacity increases.
\end{itemize}


The rest of the paper is organized as follows. Section \ref{sec2} introduces the economic dispatch problem with storage as public asset. Based on this formulation, in Section \ref{sec3}, we reexamine the definition of LMP, and propose  to decouple the LMP into constant and variant components. Section \ref{sec4} investigates the value of storage to the system as a whole as well as to the individuals. To better understand the consumers facing new market conditions, we employ the $k$-means clustering for user profiling in Section \ref{sec5}. Numerical studies verify our theoretical analysis on the value of storage in Section \ref{sec6}. Finally, we deliver the concluding remarks and point out possible future directions in Section \ref{sec7}.

\section{Problem Formulation}
\label{sec2}
In this section, we introduce the general economic dispatch problem with storage as public asset. The ISO conducts the economic dispatch over a period of interest. The key difference, compared with the conventional economic dispatch model, lies in the storage constraints. To better characterize the value of storage as public asset, we assume the ISO owns storage of total capacity $E$, and could distribute the storage in the grid at its will. To rigorously formulate this problem, we first introduce the storage constraints, then the DC approximation for the transmission line constraints, and finally the economic dispatch formulation.

\subsection{Storage Constraints}
Being public asset, the key benefit is that the ISO could distribute the storage system geographically. Specifically, denote the set of buses by $\mathcal{N}$, which contains $N\coloneq |\mathcal{N}|$ buses in the grid. Given a budget to purchase storage of total capacity $E$, the ISO could decide to install capacity $e_n$ at each bus $n\in \mathcal{N}$. This implies that
\begin{align}
    \sum_{n\in \mathcal{N}}e_n\le E. \label{eq1}
\end{align}

For the storage system at bus $n$, when conducting economic dispatch, the ISO decides its control action $u_{n,t}$ at time $t$. The action $u_{n,t}$ could be either positive (indicating charging) or negative (indicating discharging). This constructs the storage evolution constraints at each bus $n$:
\begin{align}
  x_{n,t}=x_{n,t-1}+u_{n,t},\\  
  0\le x_{n,t}\le e_n,
\end{align}
where $x_{n,t}$ denotes the state of charge (SoC) of storage at bus $n$ at time $t$. To ensure the maximal flexibility during the economic dispatch from time $0$ to time $T$, we set the terminal values of SoC both to be half of its capacity, i.e.,
\begin{align}
  x_{n,0}=x_{n,T}=\frac{e_n}{2},\  \forall n.
\end{align}
Note that, these boundary conditions also imply that within each economic dispatch cycle, there is no pure arbitrage. This also highlights the nature of public asset. 
{
\begin{remark}
    Thoughout this paper, the cost of storage is not taken
into consideration, since we want to highlight the impacts of large-scale deployments of publicly owned storage during the storage operation process.
\end{remark}
}
\subsection{Transmission Line Constraints}
The storage control actions allow us to characterize the transmission line capacity constraints. At each bus $n$, at each time $t$, we denote its generation by $g_{n,t}$ and its demand by $d_{n,t}$. Together with the storage control action $u_{n,t}$, we can calculate the net outflow $F_{n,t}$ at bus $n$:   
\begin{equation}
    F_{n,t}=g_{u,t}-u_{n,t}-d_{n,t}. \label{eq5}
\end{equation}

The DC approximation \cite{stott2009dc} for lossless transmission system states the Kirchhoff's laws in the transmission lines as follows:
\begin{equation}
    \begin{aligned}
    f_{nm,t}&=Y_{nm}(\theta_{n,t}-\theta_{m,t}),\\
     F_{n,t}&=\sum_{nm\in \mathcal{V}} f_{nm,t},
    \end{aligned}
\end{equation}
where $Y_{nm}$ is the susceptance of line $n$-$m$, $\theta_{n,t}$ denotes the phase angle at bus $n$ at time $t$, and $f_{nm,t}$ stands for the directed power flow of line $n$-$m$ at time $t$.

Hence, the transmission line capacity constraints simply require:
\begin{align}
    f_{nm,t}\le f_{nm}^{\text{max}}, \ \forall nm \in \mathcal{V}, \ \forall t,  
\end{align}
where $\mathcal{V}$ denotes the set of all transmission lines in the system. 
\subsection{Economic Dispatch with Storage}
With the aforementioned constraints, we can now formulate the economic dispatch problem with storage as public asset. Specifically, the ISO seeks to solve the following optimization problem (P1):
\begin{subequations}
  \begin{align}
    \text{(P1)} \quad \min\quad\ & \sum_{n\in\mathcal{N}}\sum\limits_{t=1}^T C_n(g_{n,t})\\
    s.t.\quad  
    &g_{n,t}-u_{n,t}-d_{n,t}=\sum_{m\in \mathcal{N}}Y_{nm}(\theta_{n,t}-\theta_{m,t}),\ \forall n,\ \forall t, \label{eq8b}\\
    &Y_{nm}(\theta_{n,t}-\theta_{m,t})\le f_{nm}^{\text{max}},\ \forall nm \in \mathcal{V},\ \forall t,\label{eq8c}\\
    &x_{n,t}=x_{n,t-1}+u_{n,t},\ \forall n,\ \forall t,\label{eq8d}\\
    &0\le x_{n,t}\le e_n,\ \forall n,\ \forall t,\label{eq8e} \\
    &x_{n,0}=\frac{e_n}{2},\ x_{n,T}=\frac{e_n}{2},\label{eq8f}\\
    &\sum_{n\in \mathcal{N}}e_n\le E. \label{eq8g}
  \end{align}
\end{subequations}
Note that $C_n(g_{n,t})$ denotes the generation cost function at bus $n$.
\begin{remark}It is possible that not all buses are connected to generators. For these degenerated buses, we can simply impose a sufficiently large cost to the corresponding generation cost function. {We choose not to include the ramping constraints in the model to highlight the role of storage. In essence, ramping constraints can be modeled as a virtual battery to provide additional flexibility. To better understand the temporal and spatial characteristics of this problem, we further assume the generation capacity for each generator is sufficiently large. Another simplification is that the
loads are assumed to be predicted perfectly, which helps us better understand how the energy storage differentially affects the prices in the system.}\footnote{In fact, without loss of generality, we can also manipulate the coefficients in the cost functions to impose the soft generation capacity constraints.}
\end{remark}

In the subsequent analysis, we adopt the quadratic cost function for analytical tractability:
\begin{assumption}
  The cost function $C_n(\cdot)$ for each bus $n$ is quadratic, i.e.,
 \begin{equation}
    C_n(g_{n,t})=\frac{1}{2}\cdot a_ng_{n,t}^2+b_ng_{n,t}+c_n,\  \forall t,\ \forall n.
 \end{equation}
\end{assumption}
{This assumption helps us examine the marginal impact of storage system at each bus neatly, which in turn enables us to better characterize the dynamics of how storage would influence different components in the system. }

\section{LMP Scheme with Storage}
\label{sec3}
The conventional LMP scheme is mostly a spatial concept. We can straightforwardly generalize the conventional definition to be the Lagrangian multipliers associated with problem (P1). However, it is important to distinguish the spatial components and the temporal components, which could enable us to better understand the value of storage systems. 
\subsection{Locational Marginal Price with Storage}
The conventional definition of LMP is defined as the shadow price for each bus $n$ for each time $t$. 
While the conventional ramping constraints already introduce certain level of temporal coupling in the short run, the integration of storage system strengthens such coupling effects across all the periods. We denote the locational price at bus $n$ at time $t$ by $p_{n,t}$. 

The closed form expression for $p_{n,t}$ can be derived from primal-dual analysis. Assigning the corresponding Lagrangian multipliers to the constraints (\ref{eq8b})-(\ref{eq8g}) in (P1), we can obtain the Lagrangian function $\mathcal{L}$ as follows:

\begin{equation}
\begin{aligned}
    \mathcal{L}=&\sum_{n\in\mathcal{N}}\sum\limits_{t=1}^T C_n(g_{n,t})+\rho\left(\sum_{n\in \mathcal{N}}e_n- E\right)\\
    &+\sum_{n\in\mathcal{N}}\sum\limits_{t=1}^T \nu_{n,t}\left[g_{n,t}-u_{n,t}-d_{n,t}\!-\!\!\!\sum_{nm\in \mathcal{V}}Y_{nm}(\theta_{n,t}-\theta_{m,t})\right]\\
    &+\sum_{nm\in\mathcal{V}}\sum\limits_{t=1}^T
    \pi_{nm,t}\left(Y_{nm}(\theta_{n,t}-\theta_{m,t})- f_{nm}^{\text{max}}\right)\\
    &+\sum_{n\in\mathcal{N}}\sum\limits_{t=2}^T
    \xi_{n,t}(x_{n,t}-x_{n,t-1}-u_{n,t})\\
    &+\sum_{n\in\mathcal{N}}\sum\limits_{t=1}^T
    \left[\lambda_{n,t}(x_{n,t}-e_n)-\mu_{n,t}x_{n,t}\right]\\
    &+\sum_{n\in\mathcal{N}}\left[\phi_{n,0}(x_{n,0}-\frac{e_n}{2})+\phi_{n,T}(x_{n,T}-\frac{e_n}{2})\right].
\end{aligned}
\end{equation}
Standard mathematical manipulations and the first order optimality conditions yield:
   \begin{subequations}
    \begin{align}
      a_ng_{n,t}^*+b_n+\nu_{n,t}^*=0,\ \forall n,\ \forall t,\label{eq12a}\\
      \sum_{nm\in\mathcal{V}}Y_{nm}(\nu_{m,t}^*-\nu_{n,t}^*)+\sum_{nm\in\mathcal{V}}\pi_{nm,t}^*Y_{nm}=0,\ \forall n,\ \forall t,\label{eq12b} \\
      -\nu_{n,t}^*-\xi_{n,t}^*=0,\ \forall n,\ \forall t,\label{eq12c}\\
      \xi_{n,t}^*-\xi_{n,t+1}^*+\lambda_{n,t}^*-\mu_{n,t}^*=0,\ \forall n,\ \forall t, \label{eq12d}\\
      -\sum_{t=1}^T \lambda_{n,t}^*+\rho^*-\frac{1}{2}\phi_{n,0}^*-\frac{1}{2}\phi_{n,T}^*=0.\ \forall n. \label{eq12e}
    \end{align}
  \end{subequations}
By rearranging (\ref{eq12a}), we have
  \begin{equation}
  \begin{aligned}
    -\nu_{n,t}^*&=a_ng_{n,t}^*+b_n\\
    &=a_n(d_{n,t}+u_{n,t}^*+F_{n,t}^*)+b_n\\
    &=(a_nd_{n,t}+b_n)+a_n(u_{n,t}^*+F_{n,t}^*)\\
    &:=p_{n,t}, \label{eq13}
  \end{aligned}
  \end{equation}
where the superscript $*$ indicates the optimal solution to the first order optimality conditions. The second equation holds due to (\ref{eq5}). The third equation indicates that LMP consists two parts. One is invariant in $E$, i.e., $a_nd_{n,t}+b_n$.
No matter how much storage is invested, all users should face such price.
We call this term \textit{constant locational marginal price} (\text{CLMP}). The other one,  $a_n(u_{n,t}+F_{n,t})$ is varying with $E$, since different storage capacity may change the optimal control actions. We call it \textit{variant locational marginal price} (\text{VLMP}). 
VLMP is affected by storage system both temporally and spatially. As the total storage capacity $E$ increases, the storage control actions will change accordingly. On the other hand, these control actions will also dramatically affect the power flow across the network ($F_{n,t}^*$'s). These two effects are coupled together and hard to distinguish.

Note that the power flow $F_{n,t}^*$ is a function of storage capacity $E$ (We will formally define such functions in Section \ref{sec4}). Using $F_{n,t}^*(E)$, we can define the conventional LMP, $p_{n,t}^{0}$:
\begin{equation}
    p_{n,t}^{0}=a_n(d_{n,t}+F_{n,t}^*(0))+b_n.
\end{equation}
Hence, compared with $p_{n,t}^0$, the storage system introduce a temporal component $a_nu_{n,t}^*$, and a spatial component $a_n(F_{n,t}^*(E)-F_{n,t}^*(0))$.
\par In fact, there is a simper way to understand the value of storage, by encoding all the temporal and spatial impact into a singe index: marginal system cost impact (MCI).

\subsection{MCI with Storage}
The MCI provides an integral treatment to examine the value of storage for individual user. Specifically, for each user $i$ at bus $n$, denote its load profile over the period of $T$ by a vector $\mathbf{L}_i=\{l_{i,n}^1,...,l_{i,n}^T\}$. We can define user $i$'s MCI over period of $T$ as follows:
\begin{equation}
  \begin{aligned}
    \text{MCI}_{i,n}&=\lim_{\delta\rightarrow 0}\frac{\sum_{t=1}^T \left(C_n\left(g_{n,t}+\frac{\delta l_{i,n}^{t}}{\Vert\mathbf{L}_{i,n}\Vert_1}\right)-C_n(g_{n,t})\right)}{\delta}\label{eq10}\\
    &=\lim_{\delta\rightarrow 0}\frac{\sum_{t=1}^T \left((a_ng_{n,t}+b_n)\cdot \frac{\delta l_{i,n}^{t}}{\Vert\mathbf{L}_{i,n}\Vert_1} + \frac{a_n}{2}\left(\frac{\delta l_{i,n}^{t}}{\Vert\mathbf{L}_{i,n}\Vert_1}\right)^2\right)}{\delta}\\
    &= \sum_{t=1}^T (a_ng_{n,t}+b_n)\cdot \frac{l_{i,n}^t}{\Vert\mathbf{L}_{i,n}\Vert_1}\\
    &=  \sum_{t=1}^T (a_nd_{n,t}+b_n) \frac{l_{i,n}^t}{\Vert\mathbf{L}_{i,n}\Vert_1}+\sum_{t=1}^T a_nu_{n,t} \frac{l_{i,n}^t}{\Vert\mathbf{L}_{i,n}\Vert_1}\\
    &\quad +\sum_{t=1}^Ta_nF_{n,t}\frac{l_{i,n}^t}{\Vert\mathbf{L}_{i,n}\Vert_1}.
  \end{aligned}
\end{equation}
\par It's clearly that $\text{MCI}$ can also be divided into two parts just as LMP: we call $ \sum_{t=1}^T (a_nd_{n,t}+b_n) \frac{l_{i,n}^t}{\Vert\mathbf{L}_{i,n}\Vert_1}$ \textit{the constant marginal system cost impact} (\text{CMCI})
and $\sum_{t=1}^T a_nu_{n,t} \frac{l_{i,n}^t}{\Vert\mathbf{L}_{i,n}\Vert_1}+\sum_{t=1}^Ta_nF_{n,t}\frac{l_{i,n}^t}{\Vert\mathbf{L}_{i,n}\Vert_1}$ \textit{the variant marginal system cost impact} (\text{VMCI}). The relationship between LMP and MCI is dictated by the following proposition.
\begin{proposition}
  For user $i$ at bus $n$, its $\text{MCI}_{i,n}$ is the weighted average electricity rate over $T$, i.e.,
  \begin{equation}
    \begin{aligned}
      \text{MCI}_{i,n}=\frac{1}{\Vert\mathbf{L}_{i,n}\Vert_1}\sum_{t=1}^T p_{n,t}\cdot l_{i,n}^t. \label{eq14}
    \end{aligned}
  \end{equation}
\end{proposition}  
This proposition makes it clear that MCI achieves the same performance as the LMP does. Hence, it enables us to understand the value of storage to the individual users via a singe index. Based on MCI, we seek to answer the following key questions: does storage benefit all users as public asset? If not, what are the key features of different types of users, in terms of their realized benefits (if any)?

\section{Value of Storage}
\label{sec4}
In this section, we examine the value of storage in terms of social cost as well as individual electricity bills. Both aspects are important for the storage to be valuable public asset. Specifically, we first use parametric analysis to highlight the social benefit of integrating storage, and then use a prototype example to demonstrate the potential issues that the storage integration may impose on individual users. This motivates us severally examine the value of storage to individuals in the electricity pool model and the general network constrained model.
\subsection{Value of Storage for Social Cost}
We define the social cost as the total generation cost in the system over period of $T$. Hence, given the storage capacity investment of $E$, the ISO can solve the optimization problem (P1) and obtain the optimal solution and the corresponding optimal objective value. Due to the quadratic cost structure assumption, the optimal objective value is unique to each capacity $E$. Hence, to evaluate the social cost, we can represent it as a function of $E$. Formally, we define a parametric function $C^*(E)$ as follows:
\begin{equation}
\begin{aligned}
        C^*(E)=\min  \    &  \sum_{n\in\mathcal{N}}\sum\limits_{t=1}^T C_n(g_{n,t}) \\
                                      s.t.\ &  \sum_{n\in \mathcal{N}}e_n\le E,  \\
            & \text{Constraints }(\ref{eq8b})\text{-}(\ref{eq8f}).
\end{aligned}
\end{equation}
This parametric function establishes the relationship between total storage capacity $E$ and the corresponding minimal generation cost. The following lemma is a direct result of Corollary 4.4.9 in \cite{Romanko04aninterior}, which states the continuity property of $C^*(E)$:
\begin{lemma}
\label{lem4.1}
  The parametric function $C^*(E)$ is continuous over $[0,+\infty)$.
\end{lemma}
\begin{remark}
In fact, the parametric function and its continuity can be extended to every parametric function defined on (P1). For example, the optimal generation $g^*_{n,t}(E)$, storage control $u_{n,t}^*(E)$ and outflow $F_{n,t}^*(E)$ are all continuous over $[0, +\infty]$. In the subsequent analysis, we directly use such notations and their continuity properties.
\end{remark}
In fact, this minimal cost function enjoys additional properties:
\begin{proposition}
\label{pro4.2}
  $C^*(E)$ is monotonically non-increasing and convex in $E$. 
\end{proposition}
The detailed proof is deferred to Appendix A. This proposition helps to identify the value of storage for social cost: 
it is not surprising to observe that a larger capacity will help improve the social welfare by reducing the total generation cost. The convexity property further eases the ISO's decision making on the optimal investment. This involves examining the amortized marginal cost for purchasing the storage systems as well as the expected marginal value of storage to the system. A detailed discussion is beyond the scope of our work.

\subsection{Motivating Example: 
Users Can Get Hurt}
\par While more storage is always beneficial to the system as a whole, it may not benefit every end user. We use a simple motivating example to highlight this fact, which will also provide us the necessary idea to investigate how the MCI's in the system evolve as capacity $E$ increases. 

Consider a two-period electricity pool model (thus, the subscript for location can be omitted). The total demands at the two periods are $d_1=10$MWh and $d_2=20$MWh, respectively. We assume a simple cost structure in the system, i.e., 
\begin{equation}
    C(g_t)=\frac{1}{2}g_t^2, \ t=1, 2.
\end{equation}
When there is no storage (i.e., $E=0$), it is straightforward to verify that:
\begin{equation}
    \begin{aligned}
        &p_{1}(0)=10\$/\text{MWh},\\
        &p_{2}(0)=20\$/\text{MWh}.
    \end{aligned}
\end{equation}
Assume there are only 2 users in the system: Alice and Bob. The load profile for Alice is $\mathbf{L}_A=(4,16)$MWh, and that for Bob is $\mathbf{L}_B=(6, 4)$MWh. These profiles allow us to determine their MCI's and the total generation cost without storage (i.e., $E=0$):
\begin{equation}
  \begin{aligned}
    \text{MCI}_A(0)=0.2\cdot10+0.8\cdot20=18\$/\text{MWh},\\
    \text{MCI}_B(0)=0.6\cdot10+0.4\cdot20=14\$/\text{MWh},\\
    C^*(0)=\frac{1}{2}\cdot (10^2+20^2)=250\$.
  \end{aligned}
\end{equation}
Suppose a storage of capacity $10$MWh is installed to improve the social welfare as a public asset, then $g_1^*(10)=g_2^*(10)=15$MWh. The prices over 2 periods are now $p_{1}(10)=p_{2}(10)=15\$/\text{MWh}$. Hence, with this storage system, we have
\begin{equation}
  \begin{aligned}
    \text{MCI}_A(10)=0.2\cdot15+0.8\cdot15=15\$/\text{MWh},\\
    \text{MCI}_B(10)=0.6\cdot15+0.4\cdot15=15\$/\text{MWh},\\
    C^*(10)=\frac{1}{2}\cdot (15^2+15^2)=225\$.
  \end{aligned}
\end{equation}
\par The social cost and MCI$_A$ are indeed reduced. However, storage does not do favor to Bob! MCI$_B$ increases, which means Bob will face a higher electricity bill. While it certainly illustrates the fact that the integration of storage may not benefit everyone, it also sheds light on how to examine the value of storage to different users: look at their load profiles!
\subsection{Electricity Pool Model}
To understand the value of storage for individual users, we first consider the electricity pool model to highlight the temporal impacts, as all the network constraints are ignored in this model. {This model can be well applied in micro-grid analysis}. The optimization problem (P1) can be simplified as follows:
\begin{subequations}
  \begin{align}
    \text{(P2)} \quad \min\quad\ & \sum\limits_{t=1}^T C(g_{t})\\
    s.t.\quad  
    &g_t-u_t=d_t,\ \forall t,\label{eq20b}\\
    &x_{t}=x_{t-1}+u_{t},\ \forall t,\label{eq20c}\\
    &0\le x_{t}\le E,\ \forall t,\label{eq20d} \\
    &x_{0}=\frac{E}{2},\ x_{T}=\frac{E}{2}.\label{eq20e}
  \end{align}
\end{subequations}
\par Clearly, Proposition \ref{pro4.2} still holds in (P2) since (P2) is a special case of (P1). Moreover, we can estimate a global lower bound for the total generation cost by Jensen's inequality:
\begin{equation}
\begin{aligned}
  C^*(E)&=   \sum\limits_{t=1}^T C(g^*_{t}(E))\\
  &\ge T\cdot C(\bar{g})=T\cdot C(\bar{d}),
\end{aligned}
\end{equation}
where $\bar{g}=\frac{1}{T}\sum_{t=1}^Tg_t$ and $\bar{d}=\frac{1}{T}\sum_{t=1}^Td_t$. The last equality is due to no pure arbitrage (i.e., $\sum_{t=1}^Tu_t=0$). This lower bound is tight when the storage capacity is sufficiently large, which forces the dispatched generations over all time slots become $\bar{g}$. At this point, the MCI for each user of any load profile becomes the same. We rigorously characterize the convergence of MCI in the following lemma.
\begin{lemma}\label{lem4.3}
  In the electricity pool model, as $E$ grows, the MCI for each user will ultimately converge to $a\bar{d}+b$.
\end{lemma}
\begin{proof}
  The Lagrangian function $\mathcal{L}$ can be formulated as follows:

\begin{equation}
\begin{aligned}
    \mathcal{L}=&\sum\limits_{t=1}^T C_n(g_{t})
    +\sum\limits_{t=1}^T \nu_{t}(g_{t}-u_{t}-d_{t})\\
    &+\sum\limits_{t=2}^T
    \xi_{t}(x_{t}-x_{t-1}-u_{t})
    +\sum\limits_{t=1}^T
    \left[\lambda_{t}(x_{t}-E)-\mu_{t}x_{t}\right]\\
    &+\phi_{0}\left(x_{0}-\frac{E}{2}\right)+\phi_{T}\left(x_{T}-\frac{E}{2}\right).
\end{aligned}
\end{equation}
  The first-order optimality conditions require:
   \begin{subequations}
    \begin{align}
      ag_{t}^*+b+\nu_{t}^*=0,\ \forall t,\label{eq23a}\\
      -\nu_{t}^*-\xi_{t}^*=0,\ \forall t,\label{eq23b}\\
      \xi_{t}^*-\xi_{t+1}^*+\lambda_{t}^*-\mu_{t}^*=0,\ \ \forall t, \label{eq23c}\\
      -\sum_{t=1}^T \lambda_{t}^*-\frac{1}{2}\phi_{0}^*-\frac{1}{2}\phi_{T}^*=0.\ \label{eq23d}
    \end{align}
  \end{subequations}
  When $E$ is sufficiently large, both LHS and RHS of (\ref{eq20d}) won't be binding at any time $t$. According to complementary slackness condition \cite{boyd2004convex}, we have $\lambda_t^*=\mu_t^*=0$. From (\ref{eq23c}), we know that $\xi_t$ will be the same for each $t$. Combining (\ref{eq23b}) with (\ref{eq23a}), we obtain $\xi_t^*=ag_t^*+b$. This implies that all $g_t^{*}$'s will be the same. Constraint (\ref{eq20e}) further requires $\sum_t^T u_t=0$. Hence $g_t^*=\bar{d}$ if $E$ is sufficiently large, which proves the proposition.
\end{proof}
While Lemma \ref{lem4.3} characterizes the MCI after convergence, it does not provide intuition on the convergent dynamics. It remains unknown whether the MCI for each user will monotonically converge to $a\bar{d}+b$, or it will oscillate around the convergent point. Through numerical observations, we find it hard to characterize the individual MCI dynamics. However, we are able to use the upper bound and lower bound of MCI to characterize the group dynamics. Specially, we can define the upper bound and the lower bound of MCI for given storage capacity $E$ as follows.
\begin{definition}
  The upper bound of MCI, UBMCI and the lower bound of MCI, LBMCI can be defined as parametric functions:
\begin{equation}
  \begin{aligned}
    \text{UBMCI}(E)=\max_i \text{MCI}_i(E), \\
    \text{LBMCI}(E)=\min_i \text{MCI}_i(E).
  \end{aligned}
\end{equation}
\end{definition}

It's straightforward to observe that UBMCI($E$) and LBMCI($E$) can be equivalently represented as follows:
\begin{equation}
  \begin{aligned}
    \text{UBMCI}(E)=ag_M^*(E)+b,\\
    \text{LBMCI}(E)=ag_m^*(E)+b,
  \end{aligned} 
\end{equation}
where $M$ and $m$ are defined as follows:
\begin{align}
  M\coloneq \arg\max_{1\le t\le T} \{g_t^*(E)\}=\arg\max_{1\le t\le T} \{d_t+u_t^*(E)\},\\
  m\coloneq \arg\min_{1\le t\le T} \{g_t^*(E)\}=\arg\min_{1\le t\le T} \{d_t+u_t^*(E)\}.
\end{align}

\begin{remark}
 UBMCI and LBMCI are obviously unique in $E$, so they can also be represented in the parametric functional forms: UBMCI$(E)$ and LBMCI$(E)$. Since $g_t^*(E)$ is continuous in $E$, UBMCI$(E)$ and LBMCI$(E)$ are also continuous in $E$.
 \end{remark}
 \par With these definitions, the following proposition characterizes the group dynamics of MCI.
\begin{proposition}
  In the electricity pool model, UBMCI$(E)$ is monotonically decreasing in $E$; LBMCI$(E)$ is monotonically increasing in $E$; and both of them converge to $a\bar{d}+b$, as $E$ approaches infinity. 
\end{proposition}
The two bounds are tight. Their monotonicities imply that a larger storage capacity can help reduce the variance of MCI, which partially indicates that more storage stabilizes the real time prices by providing more fluidity in the market. However, the monotonically increasing lower bound also indicates that more storage is not beneficial to every end user. For those who concentrate their power consumption at low-price periods, their MCI's are more likely to increase. On the contrary, for those who consume more at high-price periods, it's more possible that their MCI's will decrease with more storage in the system.

\subsection{Network Constrained Model}
\par After investigating storage integration's temporal impact on individual end users, we can now turn to the network constrained model to examine the combined temporal and spatial impacts. {This model can show the power of grid interchanges.}

While it is challenging to directly analyze the value of storage for individual users in this case, we start by examining the value of storage for each node.
\begin{proposition} 
\label{pro4.6}
  For the optimal dispatch profile given $E$, the marginal values of storage at all buses are the same. They are all non-increasing and non-negative. Mathematically,
  \begin{align} 
    \frac{\partial C^*(E)}{\partial e_1}\bigg |_{e_1=e_1^*}=...=\frac{\partial C^*(E)}{\partial e_N}\bigg|_{e_N=e_N^*}\ge 0.\label{eq26}
  \end{align}
\end{proposition}
\begin{proof}
    Rearranging the first-order condition for $e_n$, i.e., equation (\ref{eq12e}), yields that
    \begin{align}
         \sum_{t=1}^T \lambda_{n,t}^*+\frac{1}{2}\phi_{n,0}^*+\frac{1}{2}\phi_{n,T}^*=\rho^*\ge 0.\ \forall n.\label{eq30}
    \end{align}
    Note that the Lagrangian multiplier $\rho^*$ is associated with an inequality. Hence, by definition, it is non-negative. Also, the LHS of (\ref{eq30}) is exactly the marginal value of storage for each bus $i$, i.e., $ \frac{\partial C^*(E)}{\partial e_i}\bigg |_{e_i=e_i^*}$. This observation immediately leads to the main conclusion in Proposition \ref{pro4.6}.
    \par The non-increasing property is due to the convexity and non-increasing property of $C^*(E)$, as illustrated in Proposition \ref{pro4.2}.
\end{proof}
\par Next, we want to figure out the evolving dynamics of MCI in the network constrained model. Although we cannot establish the monotonicity for the MCI upper bound/lower bound across the system, we observe interesting phenomenon at each bus. Namely, as storage capacity $E$ grows, for each bus, its hourly generations across all the time slots converge to the same level. This indicates that the locational MCI also converges.
\begin{proposition}\label{pro4.7}
  In the general network constrained model, as $E$ grows, the MCI for bus $n$ will converge to $a_n\tilde{g_n}+b_n$, where $\tilde{g}_n$ is the solution to (P3):
  \begin{subequations}
    \begin{align}
      \text{(P3)} \quad \min\quad\ & \sum_{n\in\mathcal{N}} C_n(g_{n})\\
      s.t.\quad  
      &g_{n}-\frac{1}{T}\sum_{t=1}^Td_{n,t}=\sum_{m\in \mathcal{N}}Y_{nm}(\theta_{n}-\theta_{m}),\ \forall n,  \label{eq31b}\\
      &Y_{nm}(\theta_{n}-\theta_{m})\le f_{nm}^{\text{max}},\ \forall nm\in\mathcal{V}.\label{eq31c}
    \end{align}
  \end{subequations}
\end{proposition}
\par The detailed proof can be found in Appendix C. One immediate result is that locational upper and lower bounds for MCI at each bus will both converge to $a_n\tilde{g}_n+b_n$. Note the convergent values can be heterogeneous among different buses. This implies that the global upper bound will converge to $\max_{n}\{a_n\tilde{g}_n+b_n\}$ while the global lower bound will converge to $\min_{n}\{a_n\tilde{g}_n+b_n\}$.
\section{User Profiling}
\label{sec5}
To better understand the MCI dynamics for heterogeneous end users, we first conduct $k$-means clustering to identify representative end user load profiles, and then examine how their MCI's (also, CMCI's and VMCI's) vary with the total storage capacity in the system. Then, we adopt a simple  yet efficient $k$-means clustering approach to direct observing the group dynamics of MCI.

\subsection{Prototype System Setup}
We use the residential load data from Pecan Street \cite{pecanstreet}, collected from May 1 to August 9, 2015, with resolution of 1 hour.  
\par We consider the MCI dynamics in the three tier prototype system (also, this is a pool model). This prototype corresponds to the ToU pricing scheme in practice, with the off peak period (hour 0-8), peak period (hour 9-12), and partial peak (hour 12-23). We want to emphasize that there are key differences between our prototype three tier system and ToU price: the prices in our system are determined by the market conditions and will be affected by the total load in real time, whereas the ToU scheme  often offers fixed rates for the three periods.
\par The total loads in the three periods are respectively $4$MWh, $12$MWh, and $6$MWh. We assume the cost function is simply $C(g_t)=g_t^2$. This allows us to characterize the price dynamics as $E$ grows. Figure \ref{fig 2} plots the sample prices for four values of $E$: $0$MWh, $15$MWh, $30$MWh and $45$MWh. {The prices at peak, off peak and partial peak hours are respectively tagged as $p_L$, $p_H$ and $p_M$}. As expected, when $E$ is sufficiently large (in our case, $45$MWh), the prices over all the periods become the same.
\begin{figure}[htbp]
        \centering  
        \includegraphics[width=0.95\linewidth]{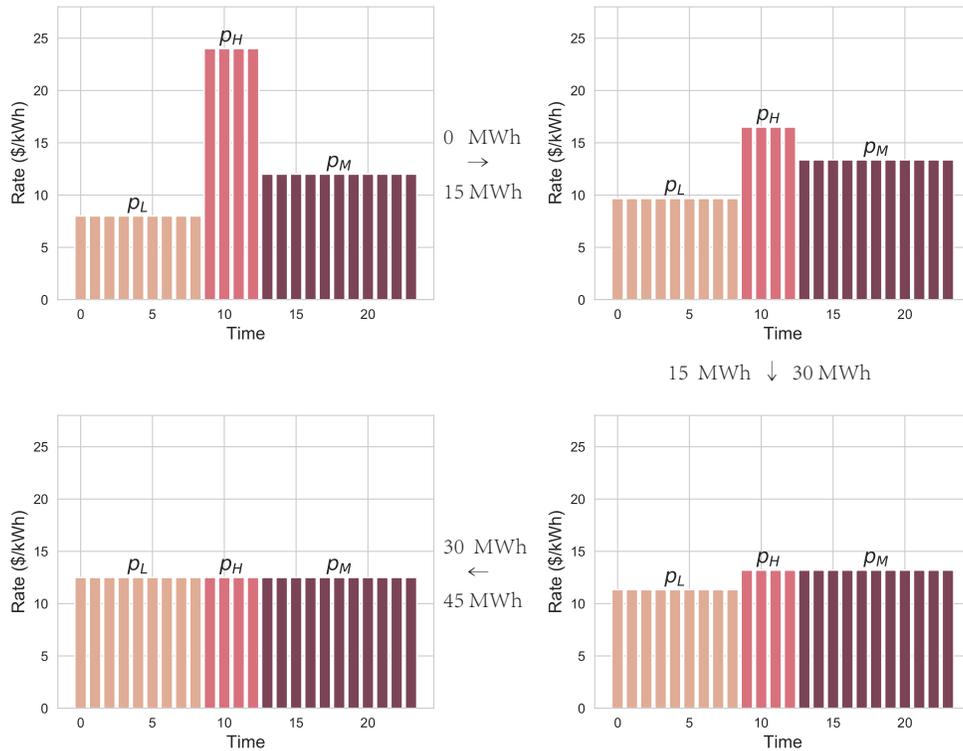}
        \caption{{Evolution of ToU Prices when Storage Capacity Goes from $0$MWh to $45$MWh. }}
        \label{fig 2}
        \vspace{-0.1cm}
\end{figure}
\begin{figure}[htbp]
        \centering  
        \includegraphics[width=0.95\linewidth]{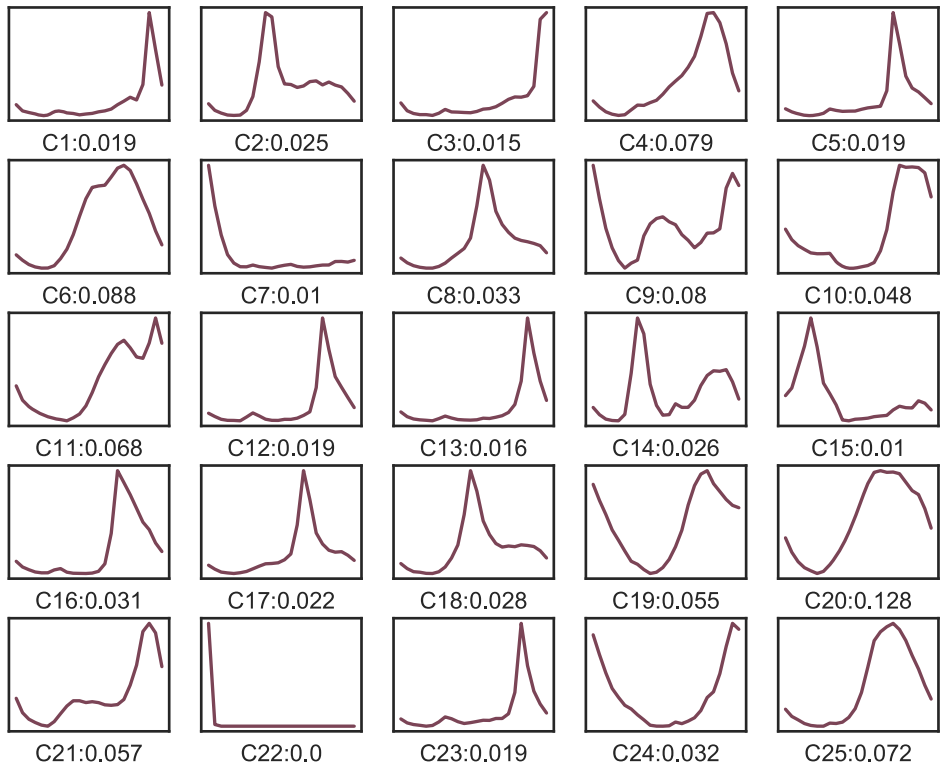}
        \caption{Clustered User Load Types: (C$X:q$) represents the proportion $q$ for cluster $X$.}
        \vspace{-0.1cm}
        \label{fig 1}
\end{figure}
\begin{figure}[htbp]
        \centering  
        \vspace{-0.15cm}
        \includegraphics[width=0.95\linewidth]{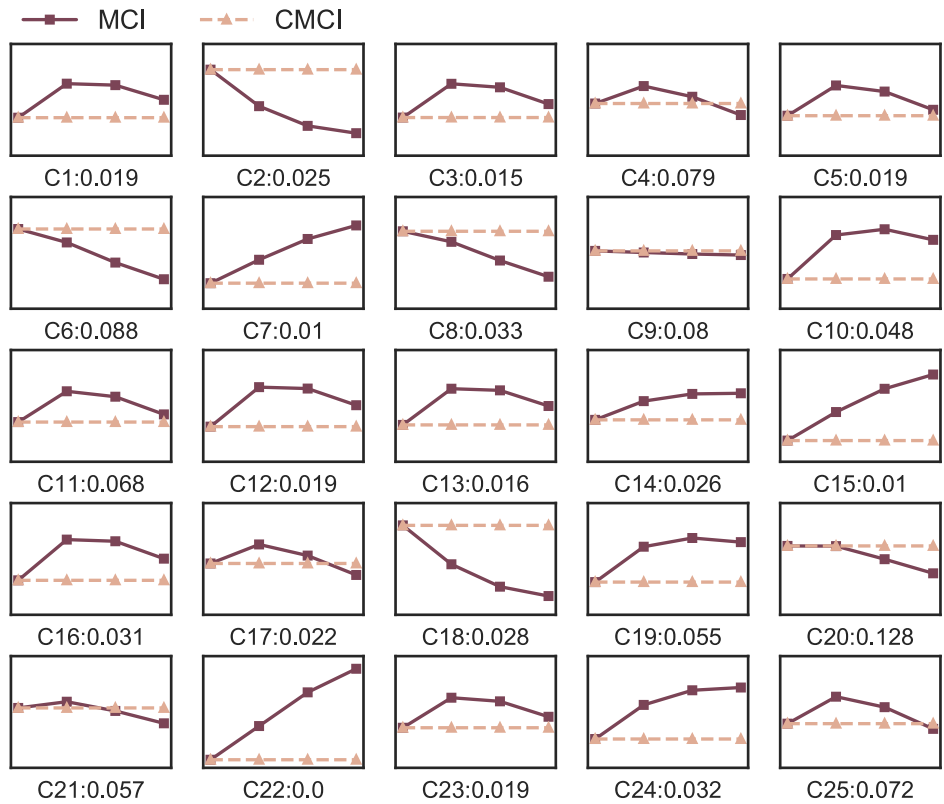}
        \caption{MCI Dynamics of Clusters (Representative Users).}
        \vspace{-0.2cm}
        \label{fig 3}
\end{figure}


\subsection{MCI Dynamics for Representative Users}
We adopt the classical $k$-means clustering method to select representative users, and set $k$ to be $25$. The clustering is based on user's normalized load profile:
\begin{align}
  \mathbf{l}_{i}=\left\{\frac{l_{i}^1}{\Vert \mathbf{L}_{i}\Vert},...,\frac{l_{i}^T}{\Vert \mathbf{L}_{i}\Vert}\right\}.
\end{align}

The clustering result is shown in Figure \ref{fig 1}. Based on this result we show the trend of MCI (decomposed as CMCI, yellow dash lines and CMCI+VMCI, red solid line) of each representative user in Figure \ref{fig 3}. {Since the CMCI is constant, it can be seen as a baseline, reflecting the variance of VMCI }. One direct conclusion is that CMCI is just the MCI when $E=0$, and VMCI can be regarded as the deviation from CMCI when $E$ grows. It can be seen that the users’ electricity consumption behaviors are quite heterogeneous. For example, type C7 users tend to consume
electricity at midnight, which result in low MCI, because the low
electricity price at midnight. On the contrary, C2 users concentrate
their consumption in the forenoon, when the price is high. Clearly,
the heterogeneity of their MCI comes from the volatility of prices. While the storage system smoothes the prices across time, its impact on individuals diverges. For instance, C$6$'s VMCI is monotonically decreasing whist C$7$'s VMCI is monotonically increasing. However, for some types of users, such as C$17$, their VMCI's increase at first and then decrease to a lower level compared with their CMCI.

\subsection{MCI Group Dynamics}
To capture the MCI group dynamics, we can sure start from the load profile based clustering result. However, it turns out that there exists a much easier algorithm. The key is to identify that it suffices to conduct the $k$-means clustering for a single metric MCI to understand its group dynamics. The $k$-means clustering based on single metric can be implemented by a greedy yet effective algorithm (the \textit{greedy $k$-means clustering}, proposed in \cite{cui2019robust}). The idea is simply to first sort the MCI's and then greedily cluster users within some prefixed radius. We repeat the algorithm in Algorithm \ref{Alg 1}.
\renewcommand{\algorithmicrequire}{\textbf{Input:}} 
\renewcommand{\algorithmicensure}{\textbf{Output:}} 
\begin{algorithm} \caption{Greedy $k$-means Clustering \cite{cui2019robust}}
\label{Alg 1}
  \begin{algorithmic}[1]
  \Require Tuple of users' MCI ($i$,MCI$_i$), $i=1,2,...,O$; Radius $r$
  \Ensure  Clusters C$1$,...,C$\kappa$
  \State Sort ($i$,MCI$_i$) by ascending order of MCI$_i$
  \State $i\leftarrow 1$, $k\leftarrow 1$  
  \Repeat
  \State $j\leftarrow \arg\max_j\{\text{MCI}_j\le \text{MCI}_i+r\}$
  \State C$_k\leftarrow \{i,...,j\}$
  \State $k\leftarrow k+1$ 
  \State $i\leftarrow j+1$
  \Until {$i>n$}
  \State $\kappa\leftarrow k$\\
  \Return Clusters C$1$,...,C$\kappa$
  \end{algorithmic}
\end{algorithm}

\begin{figure}[t]
        \centering  
        \includegraphics[width=0.84\linewidth]{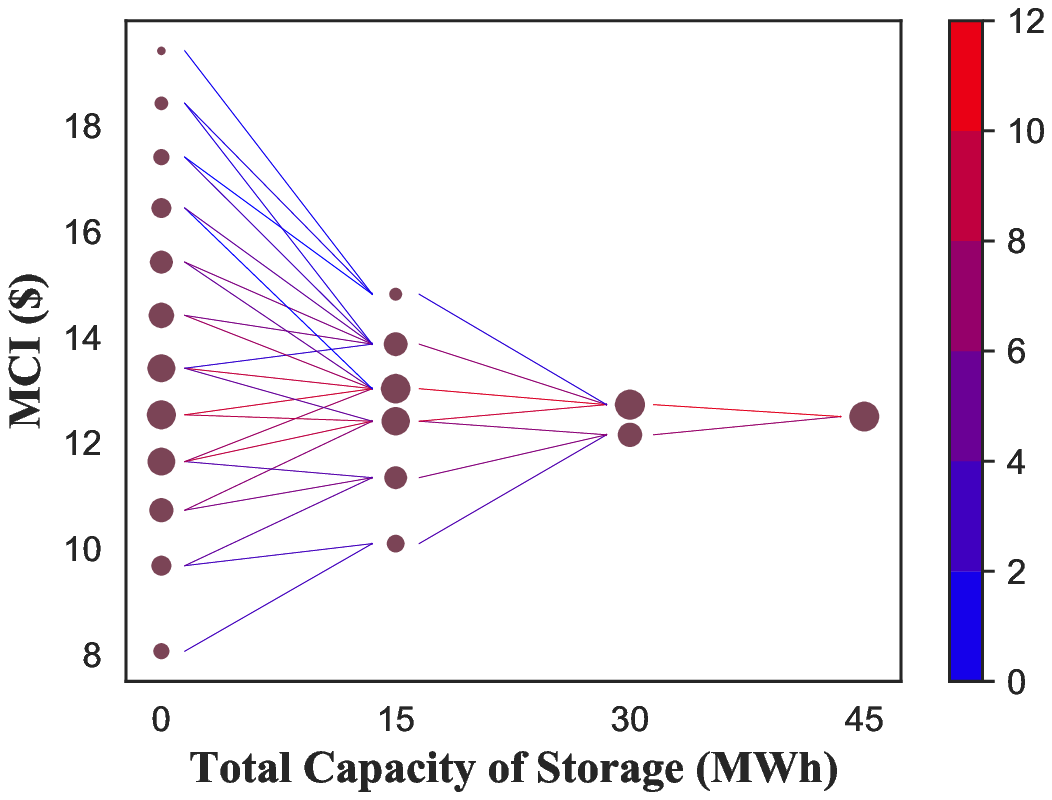}
        \vspace{-0.15cm}
        \caption{MCI Group Dynamics.}
        \vspace{-0.2cm}
        \label{fig GkC}
\end{figure}
\par This algorithm is effective as it achieves the optimal $k$-means clustering, yet with the time complexity of $O(n\log n)$. Figure \ref{fig GkC} visualizes the group dynamics of MCI: the radius of each circle indicates the number of users in the corresponding cluster and the ties characterize cluster flow dynamics. It is clear that as $E$ grows, the number of clusters decreases dramatically, and the upper and lower bounds of MCI in the system also converge very fast.

\section{Numerical Studies}
\label{sec6}
In this section, to support our theoretical results for the network constrained model, we conduct numerical studies in two systems. We first consider a 3-bus prototype system to highlight the convergence feature of MCI. Then we turn to the more realistic case: IEEE 39-bus system \cite{zimmerman2010matpower}. We find our theoretic result is valid in both cases. 


\begin{figure}[h!b]
    \centering
        \subfloat[3-bus System.]{
      \includegraphics[width=0.45\linewidth]{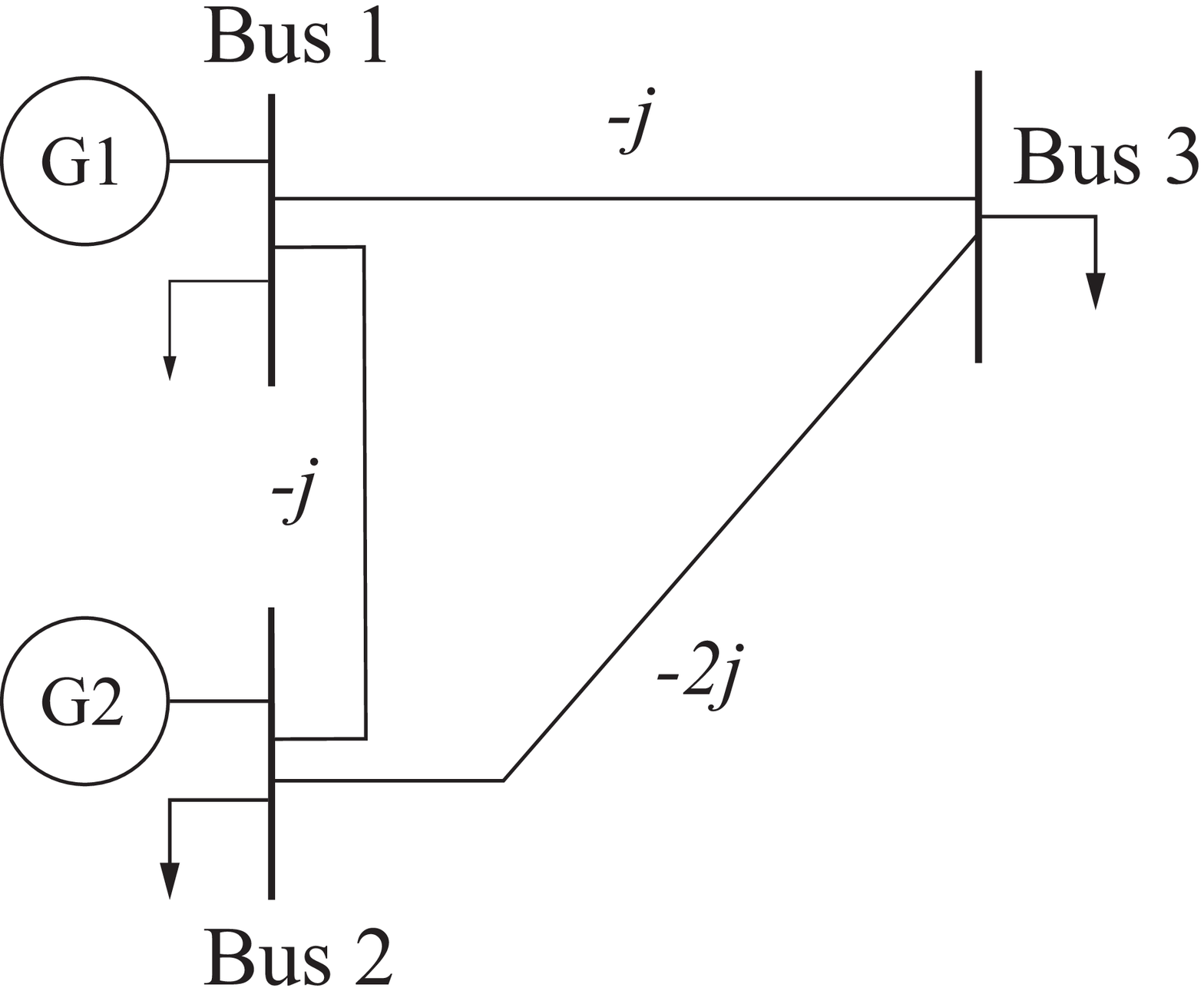}}
    \label{fig 6a}\hfill
	  \subfloat[Load Pattern.]{
        \includegraphics[width=0.48\linewidth]{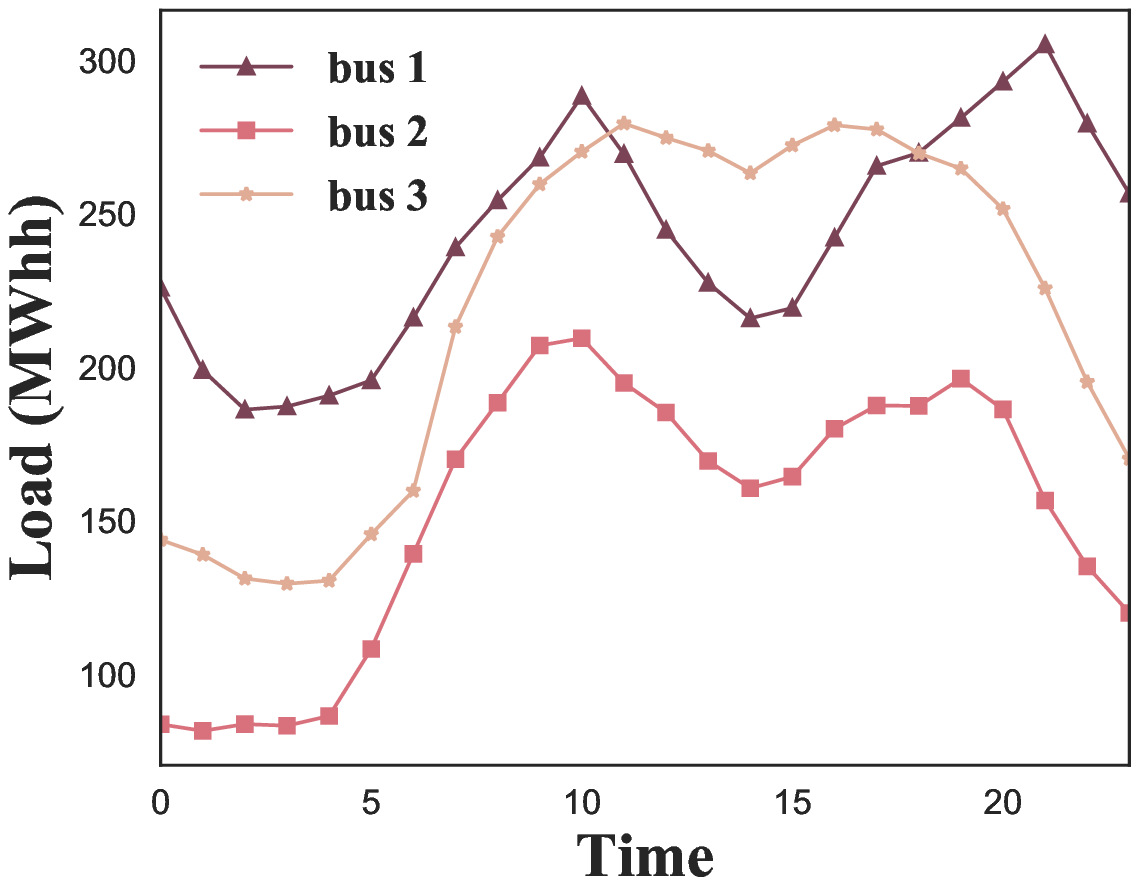}}
    \label{fig 6b}
	  \caption{Network and Load for 3-bus Prototype System.}
	  \label{fig 6} 
     \centering
	  \subfloat[Cost and Generation.]{
      \includegraphics[width=0.52\linewidth]{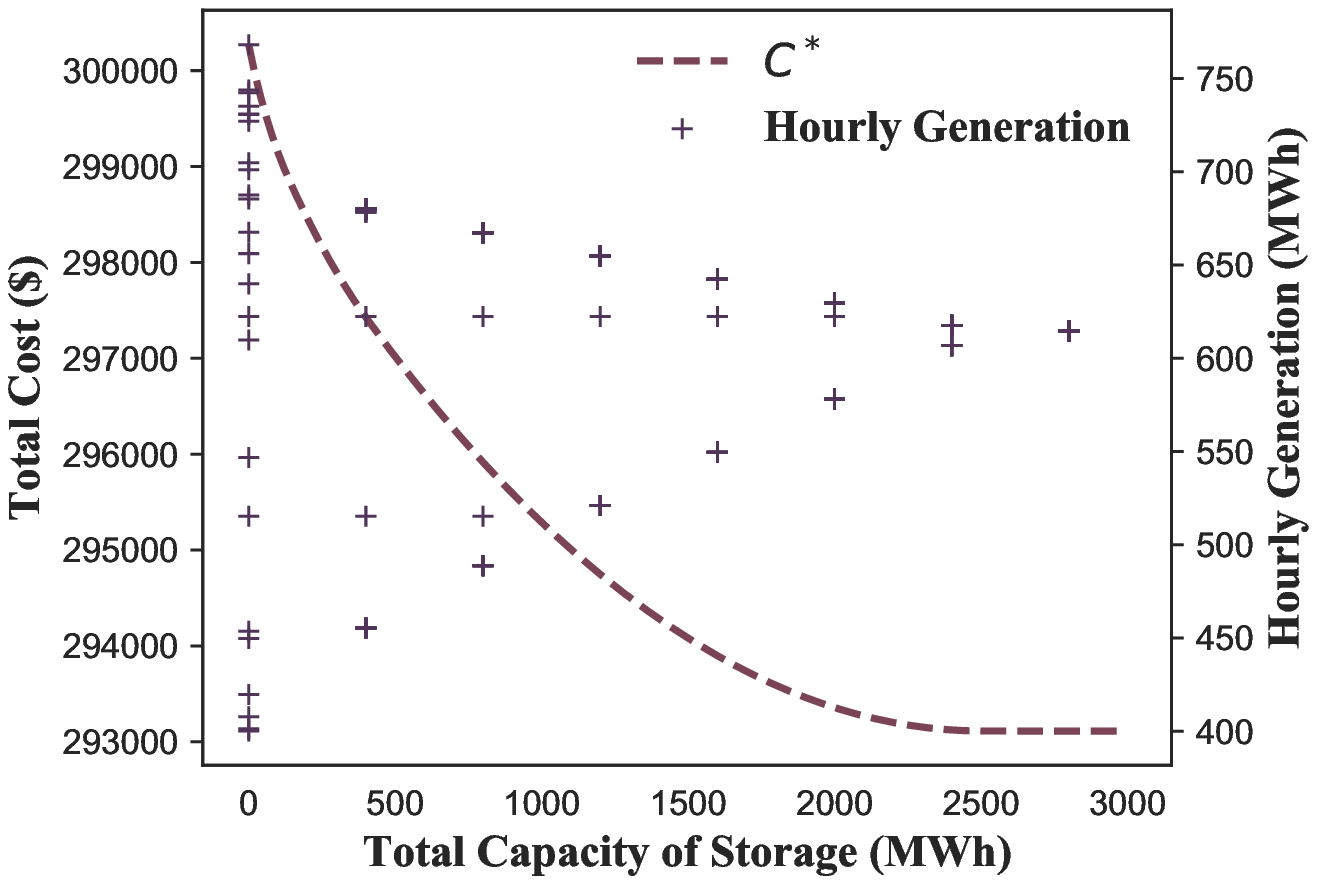}}
    \label{fig 7a}\hfill
	  \subfloat[MCI and Upper/Lower Bounds.]{
        \includegraphics[width=0.44\linewidth]{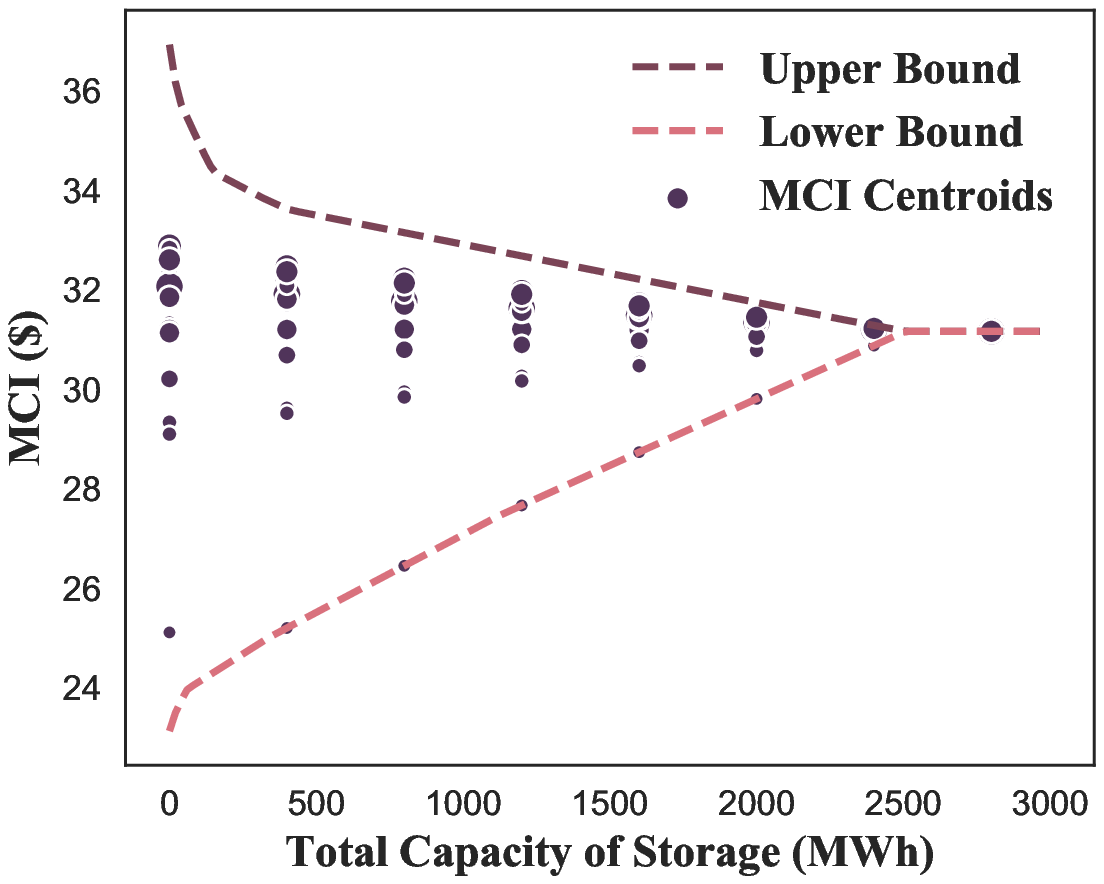}}
    \label{fig 7b}\hfill
	  \caption{Results for 3-bus Pool Model.}
	  \label{fig 7} 
	  \vspace{-0.3cm}
\end{figure}

\subsection{3-bus Prototype System}
We illustrate our results using a prototype 3-bus system. The network and the system load profiles are shown in Figure \ref{fig 6}. Both bus 1 and 2 have one generator and bus 3 is a pure load bus. The susceptance for each transmission line is shown in Figure \ref{fig 6}(a). We assume the generation cost functions are
\begin{equation}
\begin{aligned}
    C_{1}(g_{1,t})=0.05 g_{1,t}^2+5g_{1,t}+100,\ \forall t,\\
    C_{2}(g_{2,t})=0.03 g_{2,t}^2+10g_{2,t}+120,\ \forall t.
\end{aligned}
\end{equation}
The transmission line capacities are $f^{\text{max}}_{12}=f^{\text{max}}_{21}=80$MW, $f^{\text{max}}_{13}=f^{\text{max}}_{31}=130$MW and $f^{\text{max}}_{23}=f^{\text{max}}_{32}=150$MW. 

\par We verify both our results on electricity pool model and those on the network constrained model. Figure \ref{fig 7} plots the convergent dynamics in the 3-bus system ignoring all the network constraints. Clearly, the total cost function is decreasing and convex in $E$, while the hourly generations of the two generators and MCI (as well as the upper bound and lower bound) converges as $E$ grows. It is interesting to observe in Figure \ref{fig 7}(b) that the initial smoothing effect of storage system is rather strong. This indicates that the smoothing effect is mostly significant when the system most urgently needs flexibility.


\begin{figure*}[htbp]
    \centering
	  \subfloat[Bus 1.]{
      \includegraphics[width=0.32\linewidth]{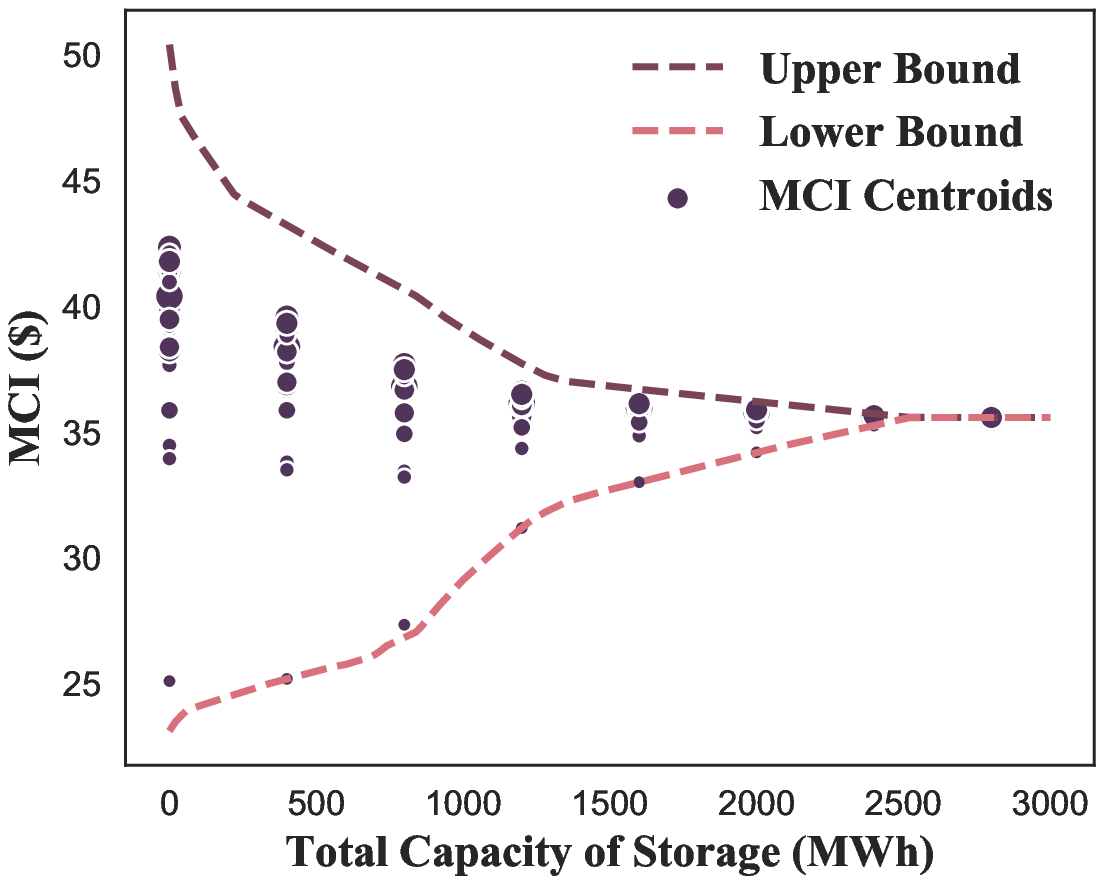}}
    \label{8a}\hfill
	  \subfloat[Bus 2.]{
        \includegraphics[width=0.32\linewidth]{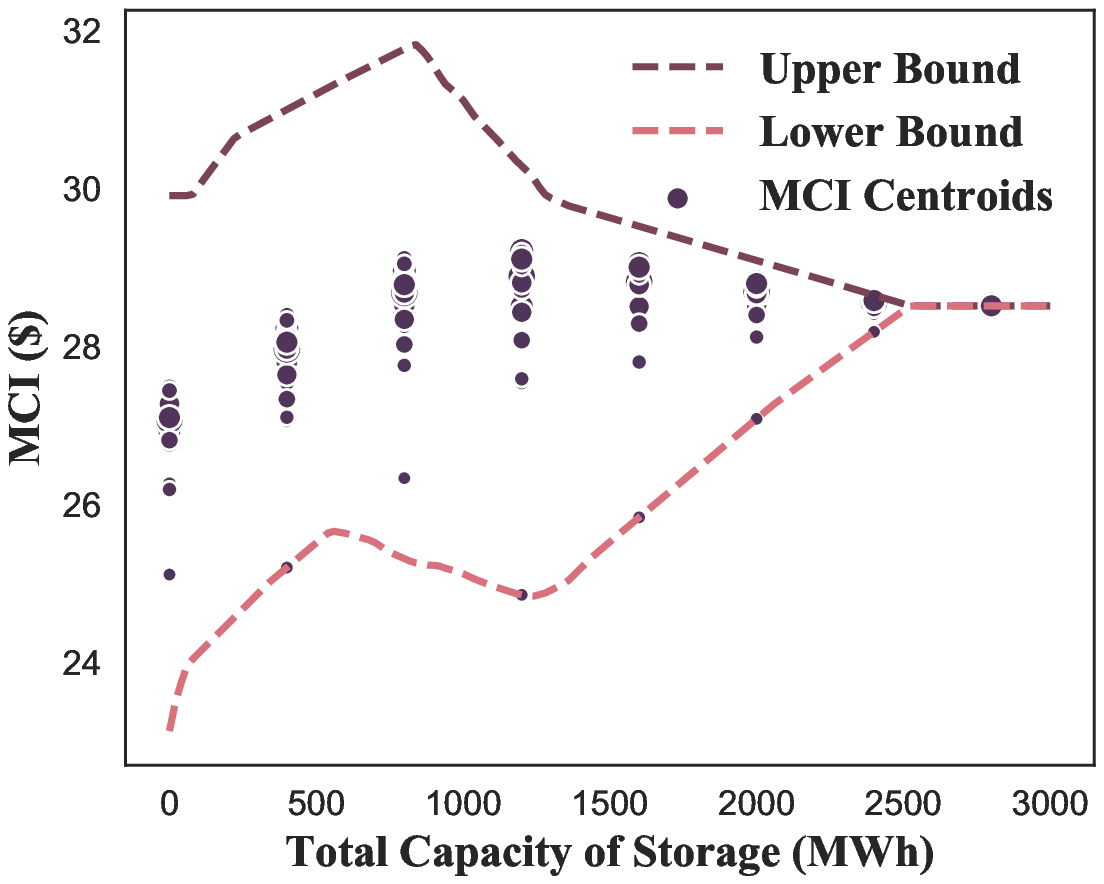}}
    \label{8b}
	  \subfloat[Bus 3.]{
        \includegraphics[width=0.32\linewidth]{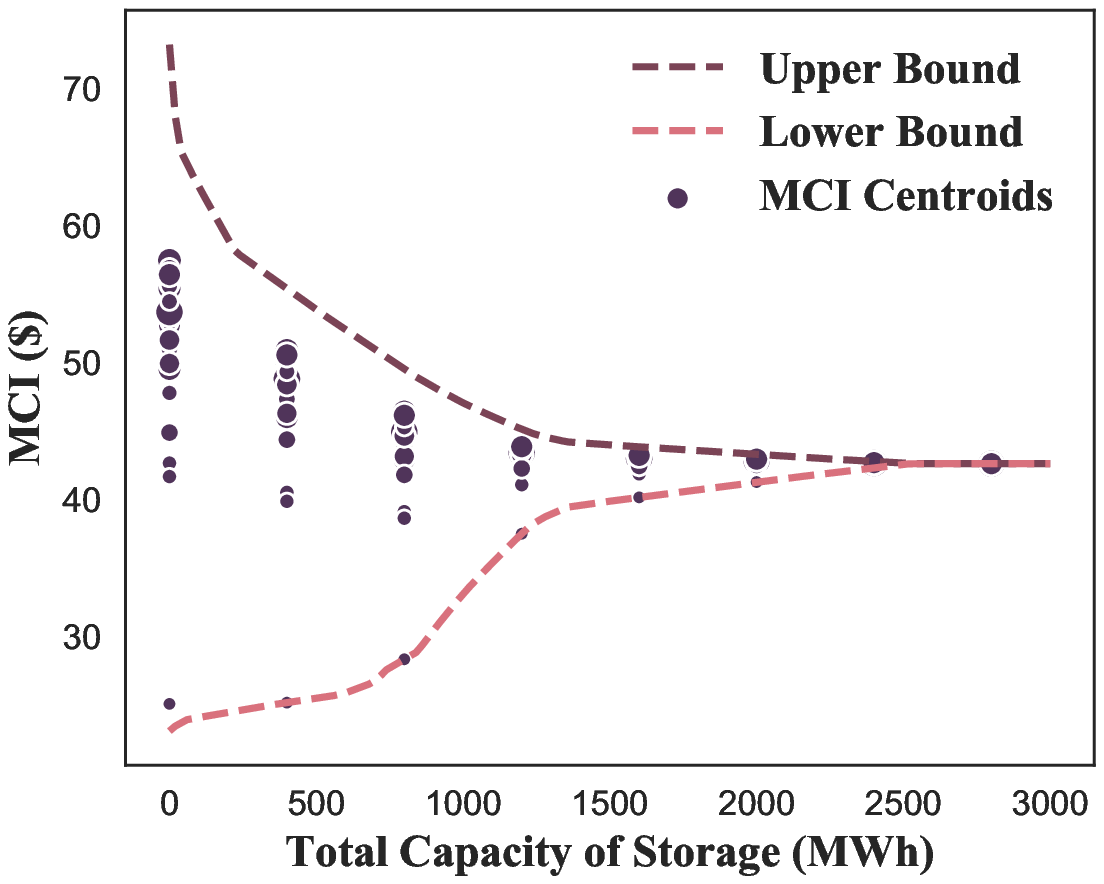}}
    \label{8c}\hfill
	  \caption{ MCI and Upper/Lower Bounds v.s. Total Storage Capacity (3-bus Network Constrained Model).}
	  \label{fig 8} 
 \end{figure*}

\begin{figure}[htbp]
\vspace{-0.2cm}
    \centering
	  \subfloat[Cost.]{
      \includegraphics[width=0.48\linewidth]{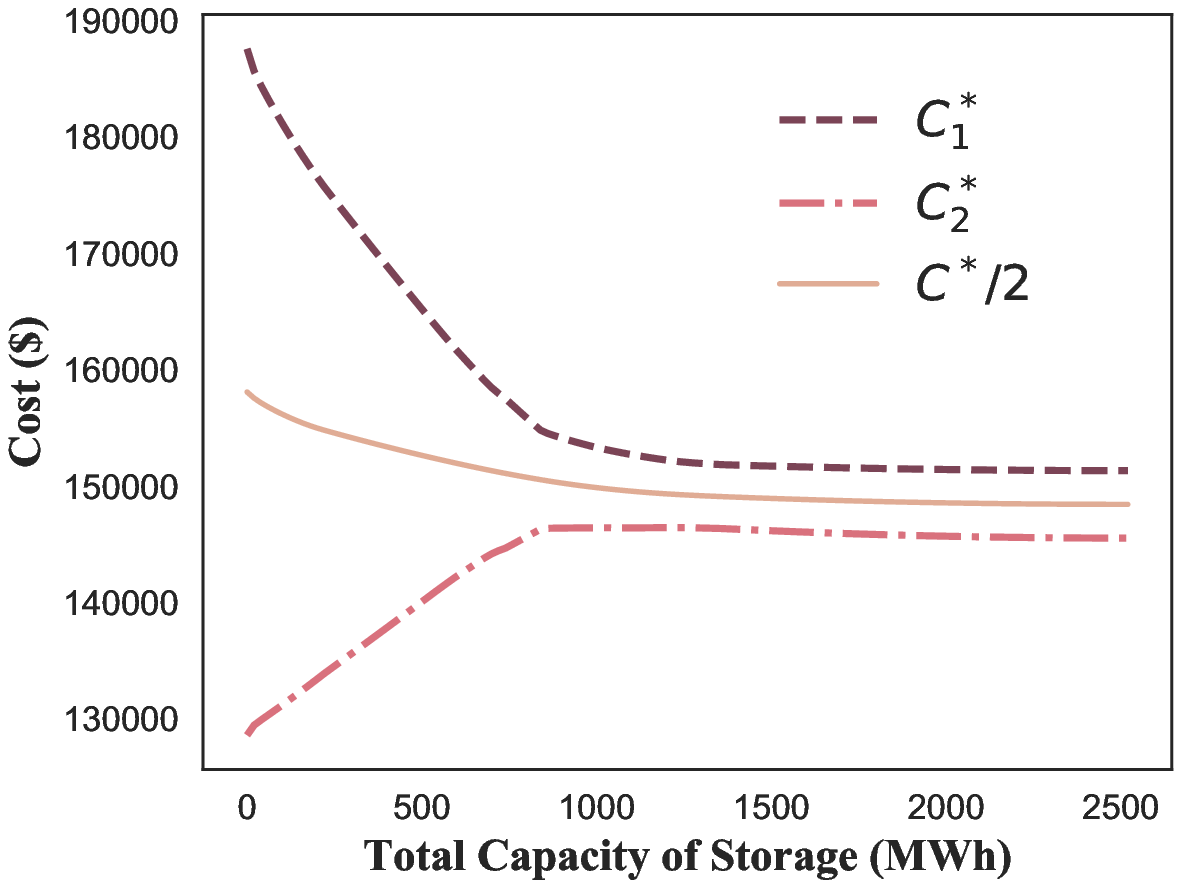}}
    \label{fig 9a}\hfill
	  \subfloat[Locational Storage Capacity.]{
        \includegraphics[width=0.465\linewidth]{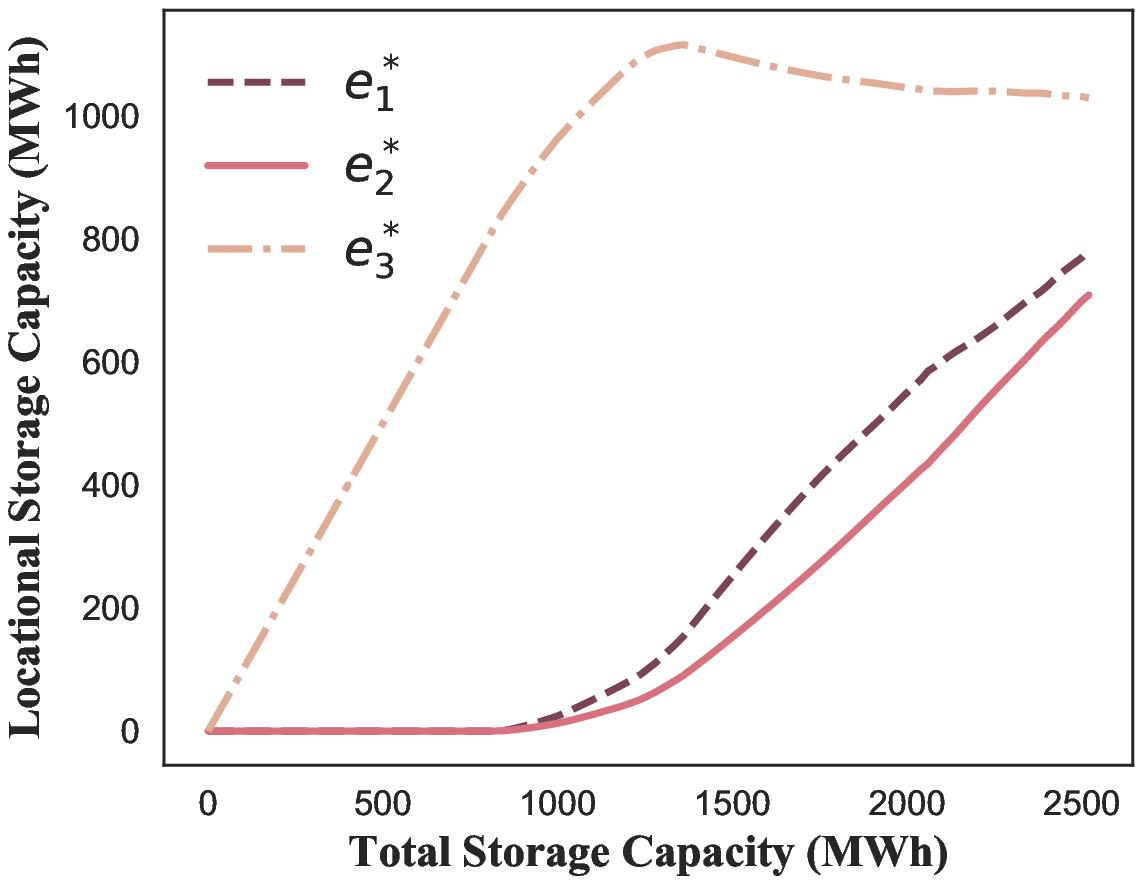}}
    \label{fig 9b}\hfill
	  \caption{Cost and Locational Storage Capacity v.s. Total Storage Capacity (3-bus Network Constrained Model).}
	  \label{fig 9} 
\end{figure}
\par With the network constraints, the convergent characteristics for MCI at each bus are visualized in Figure \ref{fig 8}. Apparently, the MCI and corresponding upper and lower bounds at each bus finally converge, but not monotonically. The result has such implication: with total storage capacity growing, the peak/off-peak generation at one bus may become larger/lower. This counter-intuitive change may help other buses to lower their costs, and ultimately leads to a lower total generation cost. Although the total cost drops, the LMP mostly relies on the cost of the local generation, which yields a even larger/lower bound for MCI.

\par Figure \ref{fig 9} depicts the cost and locational storage capacity's trends in $E$. The result shows again that storage has diverse impacts to different buses. The total generation cost is still convex and decreasing in total storage capacity. However, from a separate view, generation for some bus even increases. It's not surprising that not all the locational storage capacities are monotonically increasing in $E$. Though a larger storage capacity only has direct influence on ISO's storage control, the change in control actions will force flows between buses to change, which in turn retroacts storage control. Hence the optimal locational storage sizing may not exhibit convexity.   

\begin{figure}[htbp]
    \centering
	  \subfloat[Cluster 4.]{
      \includegraphics[width=0.48\linewidth]{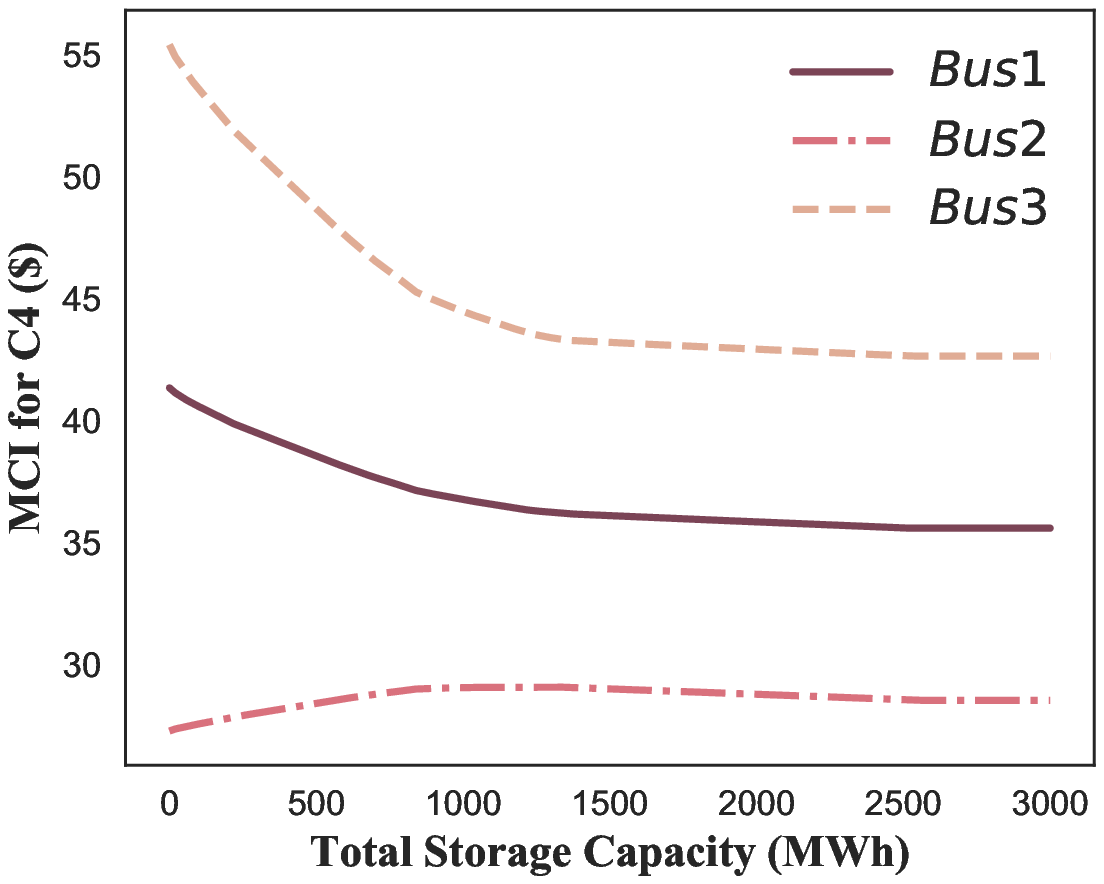}}
    \label{10a}
	  \subfloat[Cluster 8.]{
        \includegraphics[width=0.48\linewidth]{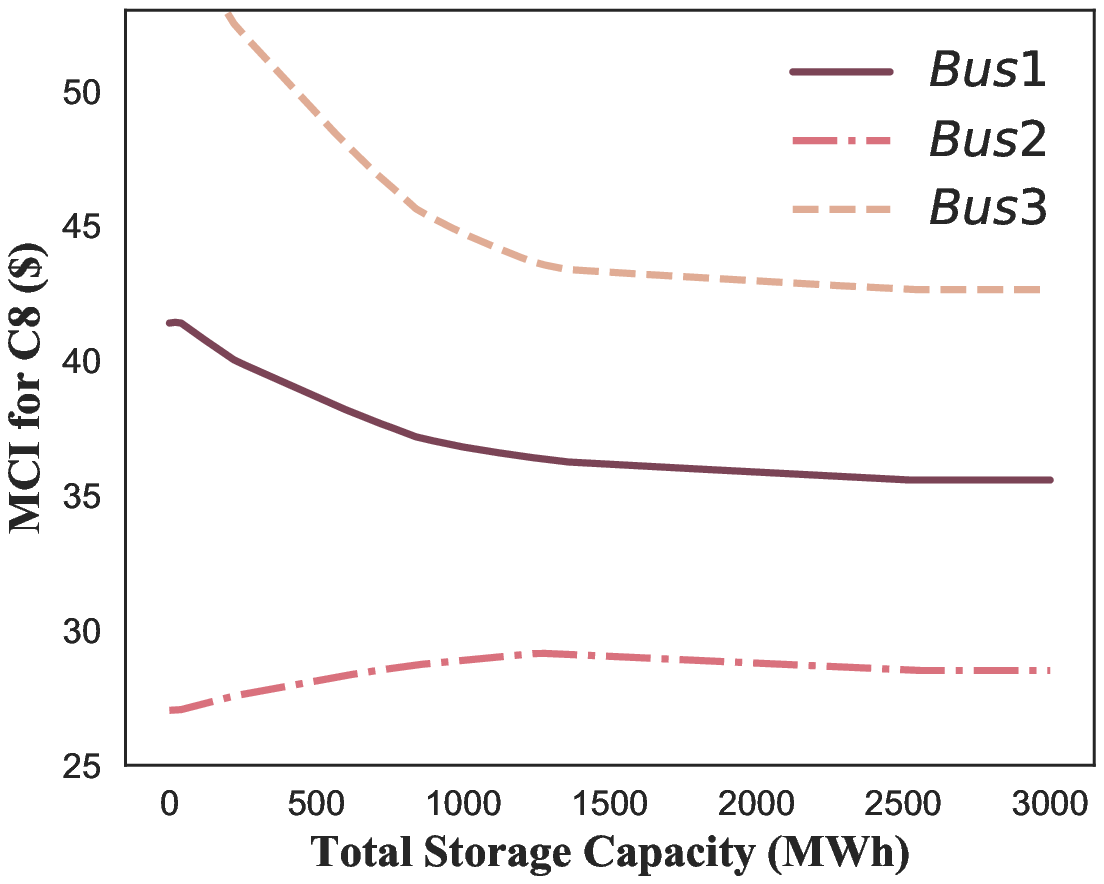}}
    \label{10b}
    
	  \subfloat[Cluster 15.]{
        \includegraphics[width=0.48\linewidth]{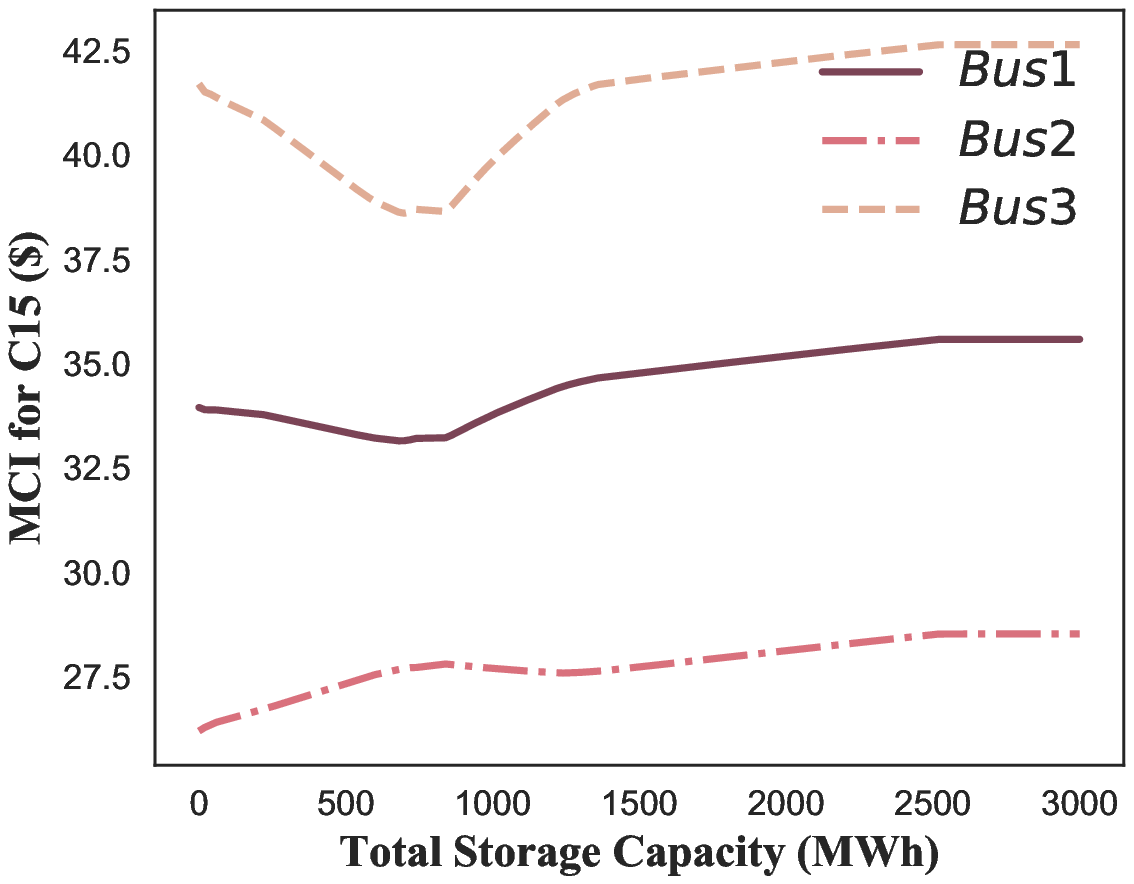}}
    \label{10c}
      \subfloat[Cluster 22.]{
      \includegraphics[width=0.48\linewidth]{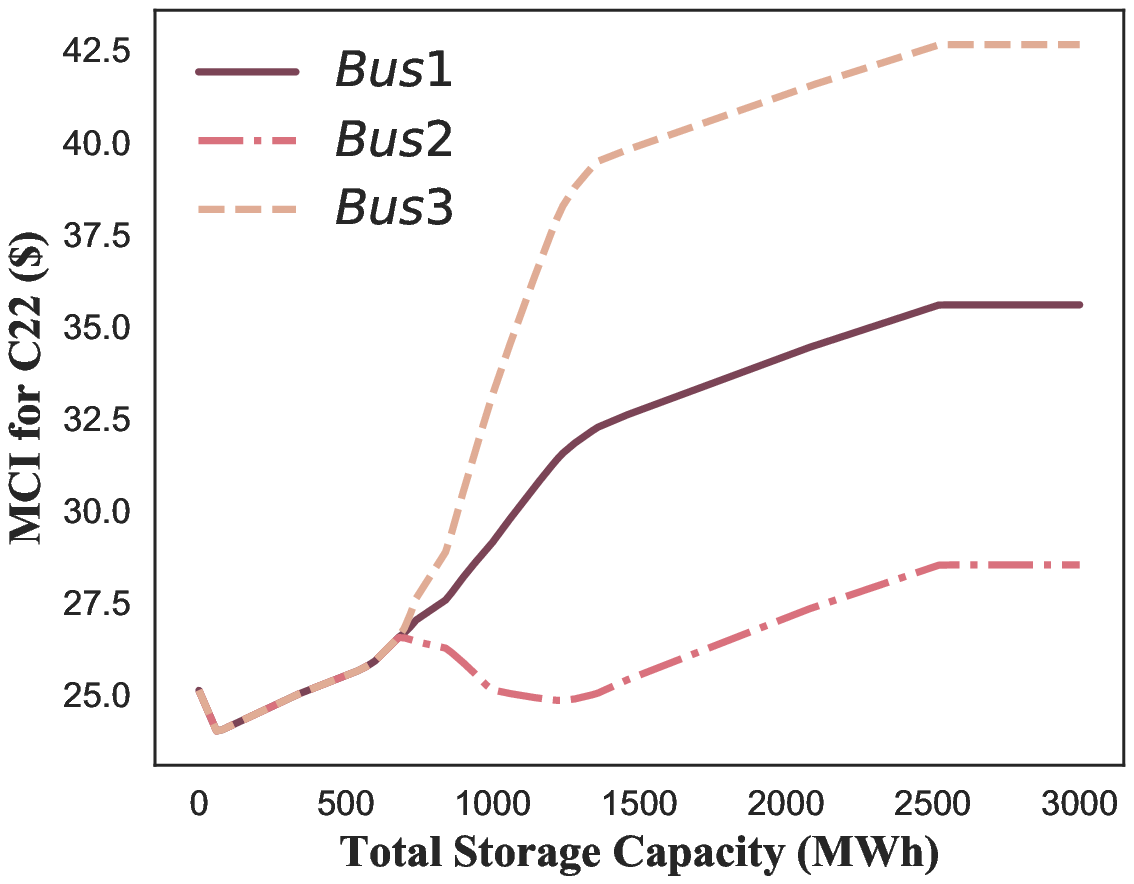}}
    \label{10d}
	  \caption{MCI for Different User Profiles (3-bus Network Constrained Model).}
	  \label{fig 10} 
\end{figure}  

\par We also highlight that a user's load profile can exhibit different MCI among buses. Figure \ref{fig 10} shows the MCI's trends for cluster C$4$, C$8$, C$15$ and C$22$. From this figure, we can see that different user profiles exhibit different patterns among the buses. Nevertheless, they all converge, as Proposition \ref{pro4.7} indicates. It's notable that for C$22$, the MCI at each bus is homogeneous when storage capacity is low. This is because such kind of users focus their power consumption when congestion doesn't occur. With storage capacity increasing, the temporal generation changes. As a result, congestion happens. The result shows that storage sometimes may not help to mitigate congestion, on the contrary, congestion conditions may be exacerbated. 

\subsection{IEEE 39-bus System}
\par In order to obtain more convincing results, we conduct the analysis on the IEEE 39-bus test system \cite{zimmerman2010matpower}. This system contains 10 generation buses. Since only single-period load is provided, we generate multi-stage load profiles by properly scaling the load. The load patterns are from the European Network of Transmission System Operators for Electricity (ENTSO-E) data \cite{entsoe}. To highlight the influences of complicated network, we only show the general case with transmission congestion in the 39-bus system.
\begin{figure}[t]
        \centering  
        \includegraphics[width=1\linewidth]{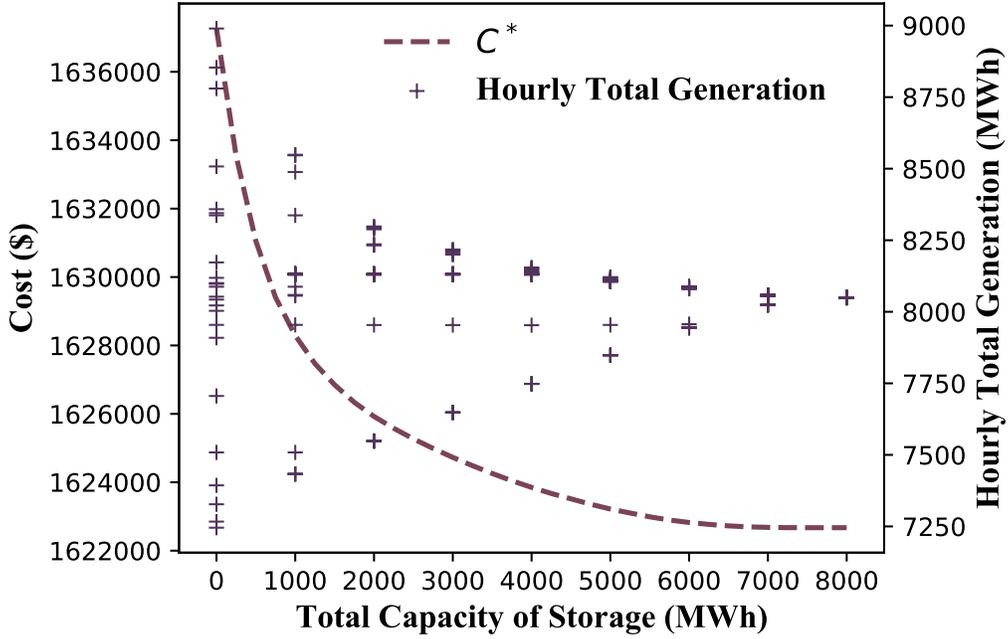}
        \caption{Total Cost and Generation v.s. Total Storage Capacity (39-bus System).}
        \label{fig 11}
\end{figure}
\begin{figure}[t]
        \centering  
        \includegraphics[width=1\linewidth]{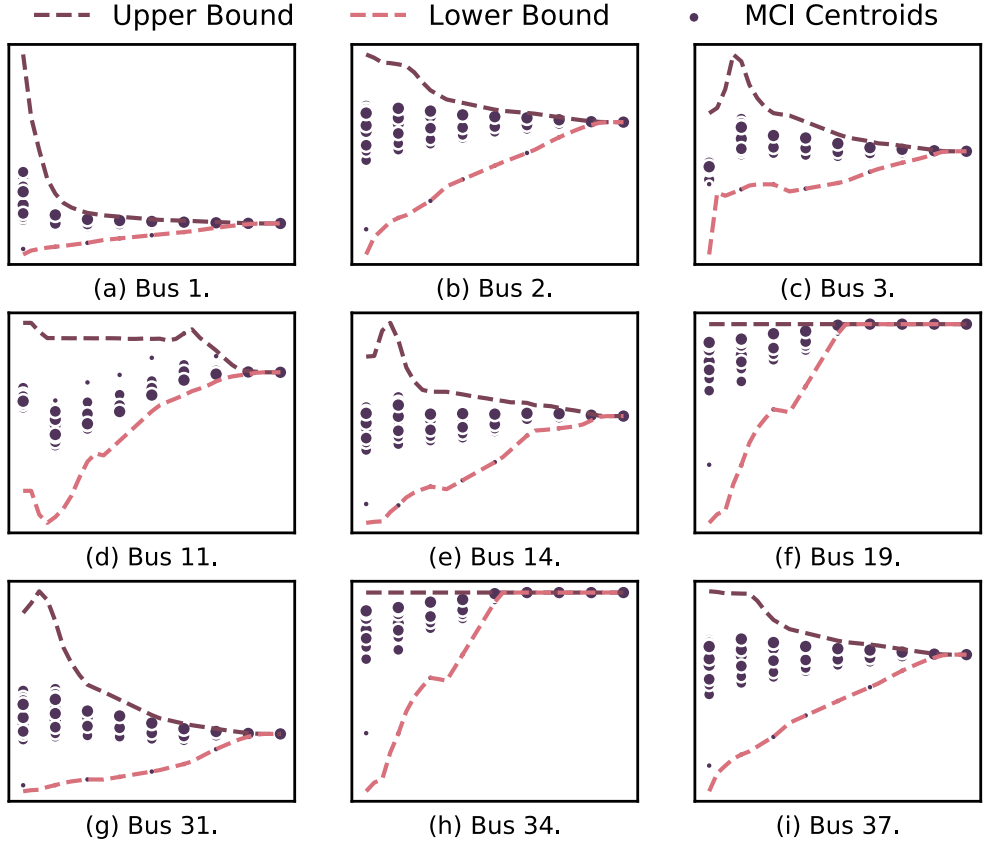}
        \caption{MCI and Upper/Lower Bounds (39-bus System).}
        \label{fig 12}
\end{figure}

\par We first show how total cost and total hourly generation change with respect to total storage capacity in Figure \ref{fig 11}. It is not surprising that the total cost is convex and decreasing in $E$. The hourly generations also converge when storage capacity increases. Note the decline of cost and convergence of generation are synchronized, i.e., when the cost curve becomes flat, the generation finally converges. 
\par The convergent characteristics of selected buses are shown in Figure \ref{fig 12}. We select nine typical buses with distinguished features. The former six cases are load buses and the latter three are generation buses. This figure exhibits variant convergent characteristics of the MCI. There is no evident difference between generation buses and load buses. We find an interesting phenomena when examining the upper and lower bounds: for some buses such as bus 19 and bus 34, the speed of convergence is faster then others. This phenomenon may come from the low variance of the demands at such buses. Also, some buses exhibit similar convergent dynamics, such as bus 2 and bus 37. This observation is due to low congestion between them. Such buses can be seen as a micro-grid as a whole.

\section{Conclusion Remarks}
\label{sec7}
In this paper, we investigate the impacts of storage as public asset to the electricity sector from two perspectives: social and individual. We prove that the storage system improves the social welfare. However, it does not benefit every end user. To examine individuals' welfare, we extend the notion of MCI as an index. We study the dynamics of MCI through $k$-means clustering and bound characterization, which exhibits valuable information of storage's value as public asset.
\par This paper can be extended in many interesting directions. For instance, as we observed in the numerical studies, when considering transmission congestion, the bounds at each bus and installed capacity are not monotone. Such observation needs to be explained by further theoretic analysis. {Since we assume that the generation cost is quadratic, it will be interesting to extend our results to more forms of cost functions. In addition, while storage may not necessarily benefit every user as public asset, the traditional pricing scheme may fail to reflect the marginal utility and individual rationality of all users. It is hence promising to design better pricing scheme from the cooperative game perspective to address the issue. Furthermore, it is valuable to combine the storage investment and the storage operator process as a whole. The major obstacle is to design an effective and fair cost allocation rule across the system.}

\bibliographystyle{unsrtnat}

\bibliography{elsarticle-template-num}
\appendix
\section{Proof of Proposition 4.2}
\label{appen}
The monotonically decreasing feature is due to the the fact that increase of $E$ expands the feasible region, which induces no worse total cost. To prove the convexity, we consider two arbitrary value of $E$: 
  $$
  0\le E_1<E_2.
  $$ 
  Then we have $g_{n,t}^*(E_1)$, $u_{n,t}^*(E_1)$, $f_{nm}^*(E_1)$, $e_n^*(E_1)$ and $g_{n,t}^*(E_2)$, $u_{n,t}^*(E_2)$, $f_{nm}^*(E_2)$, $e_n^*(E_2)$ are the corresponding optimal solutions to problem (P1) when $E=E_1$ and $E=E_2$. For any $E^\prime=\beta E_1+(1-\beta) E_2$, where $0\le \beta\le 1$, we can show that 
  \begin{align}
    g_{n,t}^\prime=\beta g_{n,t}^*(E_1)+(1-\beta)g_{n,t}^*(E_2),\\
    u_{n,t}^\prime=\beta u_{n,t}^*(E_1)+(1-\beta)u_{n,t}^*(E_2),\\
    f_{nm}^\prime=\beta f_{nm}^*(E_1)+(1-\beta)f_{nm}^*(E_2),\\
    e_{n}^\prime=\beta e_{n}^*(E_1)+(1-\beta)e_{n}^*(E_2)
  \end{align}
  construct a feasible solution to (P1). Note this is not necessarily the optimal solution for (P1) when $E=E^\prime$. Thus we can show that
  \begin{equation}
    \begin{aligned}
      C^*(E^\prime)&\le \sum_{n\in\mathcal{N}}\sum_{t=1}^TC(g_{n,t}^\prime)\\
      &= \sum_{n\in\mathcal{N}}\sum_{t=1}^T\bigg[\frac{1}{2}a(\beta g_{n,t}^*(E_1)+(1-\beta)g_{n,t}^*(E_2))^2\\
      &\quad +b(\beta g_{n,t}^*(E_1)+(1-\beta)g_{n,t}^*(E_2))+c\bigg]\\
      &\le \sum_{n\in\mathcal{N}}\sum_{t=1}^T \bigg[\frac{1}{2}a \beta (g_{n,t}^*(E_1)^2)+\frac{1}{2}a(1-\beta) (g_{n,t}^*(E_2)^2)\\
      &\quad + b (\beta g_{n,t}^*(E_1)+(1-\beta)g_{n,t}^*(E_2))+c\bigg]\\ 
      &= \beta C^*(E_1)+(1-\beta)C^*(E_2).
    \end{aligned}
  \end{equation} 
  The second inequality holds because of the convexity of quadratic cost function. This concludes our proof.  

\section{Proof of Proposition 4.5}
In Lemma \ref{lem4.3}, we have proven that 
\begin{align}
  g_t^*(E)=\frac{1}{T}\sum_{t=1}^T d_t=\bar{d}_t,\ \forall t,\ \forall E\ge \tilde{E},
\end{align}
where $\tilde{E}$ is a large number. In this case, the maximal and minimal generation are also $\bar{d}$. Hence,
\begin{align}
\lim_{E\to\infty}\text{UBMCI}(E)=\lim_{E\to\infty}\text{LBMCI}(E)=a\bar{d}+b.
\end{align}
\par Now we prove the monotonicity. We only prove the monotonically decreasing character of UBMCI since the proof for LBMCI follows the same routine. We prove this by contradiction. 
\par Let $\delta>0$ represent an infinitesimal perturbance. We need to verify $g_M^*(E+\delta)\le g_M^*(E)$, where $g_M^*$ is the largest temporal generation. Suppose $g_M^*(E+\delta)>g_M^*(E)$. Denote $\epsilon_t$ as the change of optimal generation at time $t$, given storage capacity changing from $E$ to $E+\delta$, i.e., 
\begin{align}
  \epsilon_t\coloneq g^*_t(E+\delta)-g^*_t(E)=u^*_t(E+\delta)-u^*_t(E).
\end{align}

\noindent Thus $\epsilon_M>0$ and $\sum_{t\neq M}\epsilon_t=-\epsilon_M<0$. Now the total cost is 
\begin{equation}
\begin{aligned}
  C^*(E+\delta)&=\frac{1}{2}a(d_M+u_M^*(E+\delta))^2+b(d_M+u_M^*(E+\delta))\\
  &\quad +\sum_{t\neq M}\frac{1}{2}a(d_t+u_t^*(E+\delta))^2+b(d_t+u_t^*(E+\delta))\\
  &=\frac{1}{2}a(d_M+u_M^*(E)+\epsilon_M)^2\\
  &\quad+\sum_{t\neq M}\frac{1}{2}a(d_M+u_t^*(E)+\epsilon_t)^2 +bd_M+\sum_{t\neq M}bd_t\\
  &=\sum_{t=1}^T[\frac{1}{2}a(d_t+u_t^*(E))^2+b(d_t+u_t^*(E))]\\
  &\quad +\frac{1}{2}a\sum_{t=1}^T\epsilon_t^2+a\sum_{t=1}^T[(d_t+u_t^*(E))\epsilon_t]\\
  &> C^*(E)+a\sum_{t=1}^T[(d_t+u_t^*(E))\epsilon_t]\\
  &\ge  C^*(E)+ a(d_M+u_M^*(E))\epsilon_M\\
  &\quad +a\sum_{t\neq M,\epsilon_t<0}[(d_t+u_t^*(E))\epsilon_t]\\
  &\ge C^*(E)+ a(d_M+u_M^*(E))\left(\epsilon_M+\sum_{t\neq M,\epsilon_t<0}\epsilon_t\right)\\
  &\ge  C^*(E).
\end{aligned}
\end{equation}
This result violates the decreasing character of $C^*(E)$, which estabilishes the contradiction. Hence, $g_M^*(E+\delta)\le g_M^*(E)$. With the continuity property (Lemma \ref{lem4.1}), the proof is completed.

\section{Proof of Proposition \ref{pro4.7}}
  First, we prove the convergence. Suppose $E$ is sufficiently large, storage capacity $e_n$ for each bus will become large enough so that (\ref{eq8e}) is not binding. As a result, $\lambda_{n,t}^*=\mu_{n,t}^*=0$. According to (\ref{eq12d}), we know $\xi_{n,t+1}^*=\xi_{n,t}^*$ for each $n$ and $t$. Namely, $\xi_{n,t}^*$ will be the same for all $t$. Adding (\ref{eq12c}) to (\ref{eq12a}), we obtain $a_ng_{n,t}^*+b_n=\xi_{n,t}^*$. Hence, $a_ng_{n,t}^*+b_n$ are the same for all time $t$. This proves the convergence.
  \par Since $g_{n,t}$ will converge for each $n$ when $E$ grows sufficiently large (denote the threshold as $E_\text{con}$), (P1) is equivalent to the following problem (P4) when $E\ge E_\text{con}$:
\begin{subequations}
\begin{align}
     (P4)\quad  \min  \    &  \sum_{n\in\mathcal{N}}\sum\limits_{t=1}^T C_n(g_{n,t}) \\
                                      s.t.\ & g_{n,1}=g_{n,2}=...=g_{n,T},\ \forall n , \label{eq43b} \\
            & \text{Constraints }(\ref{eq8b})\text{-}(\ref{eq8g}). \notag
\end{align}
\end{subequations}
Relatively summing up (\ref{eq8b}) and (\ref{eq8c}) over all $t$ and dividing them by $T$, we have
  \begin{align}
     \frac{1}{T}\sum_{t=1}^T(g_{n,t}-d_{n,t})=\frac{1}{T}\sum_{t=1}^T\sum_{m\in \mathcal{N}}Y_{nm}(\theta_{n,t}-\theta_{m,t}), \forall n,\label{eq44}\\
     \frac{1}{T}\sum_{t=1}^TY_{nm}(\theta_{n,t}-\theta_{m,t})\le f_{nm}^{\text{max}},\ \forall nm \in \mathcal{V}.
  \end{align}
  Note $u_{n,t}$'s are eliminated because $\sum_{t=1}^Tu_{n,t}=0,\ \forall n.$ Denote 
  \begin{align}
      g_n&\coloneq \frac{1}{T}\sum_{t=1}^Tg_{n,t},\ \forall n,\\ \theta_n&\coloneq\frac{1}{T}\sum_{t=1}^T\theta_{n,t},\ \forall n.
  \end{align} 
  Then we have 
  \begin{align}
  g_{n}-\frac{1}{T}\sum_{t=1}^Td_{n,t}=\sum_{m\in \mathcal{N}}Y_{nm}(\theta_{n}-\theta_{m}),\ \forall n, \\
      Y_{nm}(\theta_{n}-\theta_{m})\le f_{nm}^{\text{max}},\ \forall nm\in\mathcal{V}.
\end{align}
  These are exactly the constraints for (P3). Hence the feasible solutions to (P4) are all feasible to (P3).
  \par Then we prove the other side. Denote 
  \begin{align}
      \hat{u}_{n,t}\coloneq \frac{1}{T}\sum_{t=1}^Td_{n,t}-d_{n,t},\ \forall n.\label{eq49}
  \end{align}
  When $E$ is sufficient large, each $e_n$ can be arbitrarily large, so (\ref{eq49}) is feasible for constraints (\ref{eq8d})-(\ref{eq8g}). Suppose $\hat{g}_n,\ \hat{\theta}_n$ construct a feasible solution to (P3), it's easy to show that
  \begin{align}
      g_{n,t}=\hat{g}_{n},\ 
      \theta_{n,t}=\hat{\theta}_{n},\ 
      u_{n,t}=\hat{u}_{n,t}
  \end{align}
  construct a feasible solution to (P4). That means, all feasible solutions to (P3) are also feasible to (P4).
  \par Now we have shown that (P3) and (P4) have the same feasible domain. Since the objective functions of (P3) and (P4) are equivalent under the constraint (\ref{eq43b}) (they are proportional with a scalar $\frac{1}{T}$), the optimal solution to (P3) is also optimal for (P4). Further, when $E\ge E_\text{con}$, (P4) is equivalent to (P1). Hence, the optimal solution to (P3) is optimal for (P1), when $E$ grows sufficiently large. Q.E.D. 
  
\end{document}